\documentclass[a4paper,UKenglish,cleveref, autoref, thm-restate]{lipics-v2021}
\nolinenumbers %uncomment to disable line numbering
\pdfoutput=1 %uncomment to ensure pdflatex processing (mandatatory e.g. to submit to arXiv)
\hideLIPIcs  %uncomment to remove references to LIPIcs series (logo, DOI, ...), e.g. when preparing a pre-final version to be uploaded to arXiv or another public repository

\usepackage[utf8]{inputenc}
\usepackage{amsthm}
\usepackage{amsmath}
\usepackage{graphicx}
\usepackage{svg}

\usepackage[appendix=append,bibliography=common]{apxproof}\usepackage[append]{optional}
\usepackage{etoolbox} %for better programming facilities
\newbool{fullversion}
\booltrue{fullversion}

\usepackage{subcaption}

\usepackage[normalem]{ulem}

\usepackage[capitalise]{cleveref}

\usepackage[disable,textsize=scriptsize]{todonotes} 
\newtheoremrep{thm}{Theorem}[section]
\newtheoremrep{lem}[thm]{Lemma}
\newtheoremrep{cor}[thm]{Corollary}
\newtheorem{defn}[thm]{Definition}
\newtheorem{prop}[thm]{Proposition}

\newcommand{\starriness}{s}
\newcommand{\merginess}{m}
\newcommand{\splitsto}{\preceq}
\newcommand{\mergesto}{\succeq}
\newcommand{\melt}{\mathrm{melt}}
\newcommand{\pbasis}{\mathcal B}
\newcommand{\mtypes}{M}

\title{Thermodynamically Driven Signal Amplification} 

\author{Joshua {Petrack}}{University of California--Davis, Davis, CA, USA}{jgpetrack@ucdavis.edu}{https://orcid.org/0009-0000-9088-9437}{NSF grants 1900931 and 1844976.}

\author{David {Soloveichik}}{University of Texas at Austin, Austin, TX, USA \and \url{https://users.ece.utexas.edu/~soloveichik/}}{david.soloveichik@utexas.edu}{https://orcid.org/0000-0002-2585-4120}{NSF grant 1901025, Sloan Foundation Research Fellowship.}

\author{David {Doty}}{University of California--Davis, Davis, CA, USA \and \url{https://web.cs.ucdavis.edu/~doty/}}{doty@ucdavis.edu}
{https://orcid.org/0000-0002-3922-172X}
{NSF grants 2211793, 1900931, and 1844976.}

\authorrunning{J. Petrack, D. Soloveichik, and D. Doty} 

\Copyright{Joshua Petrack, David Soloveichik and David Doty} 



\category{} %optional, e.g. invited paper

\relatedversion{} %optional, e.g. full version hosted on arXiv, HAL, or other respository/website
\EventEditors{John Q. Open and Joan R. Access}
\EventNoEds{2}
\EventLongTitle{42nd Conference on Very Important Topics (CVIT 2016)}
\EventShortTitle{CVIT 2016}
\EventAcronym{CVIT}
\EventYear{2016}
\EventDate{December 24--27, 2016}
\EventLocation{Little Whinging, United Kingdom}
\EventLogo{}
\SeriesVolume{42}
\ArticleNo{23}
\begin{document}

\maketitle

\begin{abstract}
The field of chemical computation attempts to model computational behavior that arises when molecules, typically nucleic acids, are mixed together. By modeling this physical phenomenon at different levels of specificity, different operative computational behavior is observed. Thermodynamic binding networks (TBNs) is a highly abstracted model that focuses on which molecules are bound to each other in a ``thermodynamically stable'' sense. Stability is measured based only on how many bonds are formed and how many total complexes are in a configuration, without focusing on how molecules are binding or how they became bound.
By defocusing on kinetic processes, TBNs attempt to naturally model the long-term behavior of a mixture (i.e., its thermodynamic equilibrium). 

We study the problem of \emph{signal amplification}: detecting a small quantity of some molecule and amplifying its signal to something more easily detectable.
This problem has natural applications such as disease diagnosis.
By focusing on thermodynamically favored outcomes, we seek to design chemical systems that perform the task of signal amplification robustly without relying on kinetic pathways that can be error prone and require highly controlled conditions (e.g., PCR amplification).

It might appear that a small change in concentrations can  result in only small changes to the thermodynamic equilibrium of a molecular system.
However, we show that it is possible to design a TBN that can ``exponentially amplify'' a signal represented by a single copy of a monomer called the \emph{analyte}:
this TBN has exactly one stable state before adding the analyte and exactly one stable state afterward, and those two states ``look very different'' from each other. In particular, their difference is exponential in the number of types of molecules and their sizes. 
The system can be programmed to any desired level of resilience to false positives and false negatives. 
To prove these results, we introduce new concepts to the TBN model, particularly the notions of a TBN's entropy gap to describe how unlikely it is to be observed in an undesirable state, and feed-forward TBNs that have a strong upper bound on the number of polymers in a stable configuration.

We also show a corresponding negative result: 
a \emph{doubly} exponential upper bound, meaning that there is no TBN that can amplify a signal by an amount more than doubly exponential in the number and sizes of different molecules that comprise it.
We leave as an open question to close this gap by either proving an exponential upper bound,
or giving a construction with a doubly-exponential difference between the stable configurations before and after the analyte is added.

Our work informs the fundamental question of how a thermodynamic equilibrium can change as a result of a small change to the system (adding a single molecule copy). 
While exponential amplification is traditionally viewed as inherently a non-equilibrium phenomenon, we find that in a strong sense exponential amplification can occur at thermodynamic equilibrium as well---where the ``effect'' (e.g., fluorescence) is exponential in types and complexity of the chemical components.
\end{abstract}

% no page number on title page, start at page 1 for main text
\thispagestyle{empty}
\clearpage
\pagenumbering{arabic}

\section{Introduction}
\label{sec:intro}

Detecting a small amount of some chemical signal, or analyte, is a fundamental problem in the field of chemical computation. The current state-of-the-art in nucleic acid signal amplification is the polymerase chain reaction (PCR)\cite{schochetman1988pcr}.
% DD: I made the description of RT-PCR more accurate, but really we could call it PCR and lose the RT part for this discussion.
% After translating the RNA to DNA using reverse transcriptase,
By using a thermal cycler,
% DD: I think the reverse transcriptase only works once at the start, to translate the RNA to DNA. Then normal PCR works on the DNA, using a high-temp DNA polymerase that never gets deactivated.
% activate and deactivate a reverse transcriptase enzyme, 
PCR repeatedly doubles the amount of the DNA strand that is present.
One downside is the need for a PCR machine, which is expensive and whose operation can be time-consuming.
The advantages of PCR are that it can reliably detect even a single copy of the analyte if enough doubling steps are taken, and it is fairly (though not perfectly) robust to incorrect results.
\todo{DS: It's not clear to me that false-negatives are a bigger problem than false-positives in PCR. Is this well-known?
JP: not sure, this was a line Dave contributed :) I'm just going to remove it since it's not especially important to any point we're making.}
Recent work in DNA nanotechnology achieves ``signal amplification'' through other kinetic processes involving pure (enzyme-free) DNA systems, such as hybridization chain reaction (HCR)~\cite{choi2018third}, classification models implemented with DNA~\cite{lopez2018molecular}, hairpin assembly cascades~\cite{xiong2021minimizing}, and ``crisscross'' DNA assembly~\cite{minev2021robust}. 

Although highly efficacious, 
PCR and these other techniques essentially rely on \emph{kinetic} control of chemical events, and the thermodynamic equilibria of these systems are not consistent with their desired output.
Can we design a system so that, if the analyte is present, the thermodynamically most stable state of the system looks one way, and if the analyte is absent, the thermodynamically most stable state looks ``very different''
(e.g., many fluorophores have been separated from quenchers)?
Besides answering a fundamental chemistry question, 
such a system is potentially more robust to false positives and negatives.
It also can be simpler and cheaper to operate: for many systems, heating up the system and cooling it down slowly reaches the system's thermodynamic equilibrium.

We tackle this problem of signal detection in the formal model of Thermodynamic Binding Networks (TBNs)~\cite{doty2017tbn,breik2019computing}.
The TBN model of chemical computation ignores kinetic and geometric constraints in favor of focusing purely on configurations describing which molecules are bound to which other molecules. 
A TBN yields a set of stable configurations, the ways in which monomers (representing individual molecules, typically strands of DNA) are likely to be bound together in thermodynamic equilibrium. A TBN performs the task of signal amplification if its stable configurations, and thus the states in which it is likely to be observed at equilibrium, change dramatically in response to adding a single monomer. 
TBNs capture a notion of what signal amplification can look like for purely thermodynamic chemical systems, without access to a process like PCR that repeatedly changes the conditions of what is thermodynamically favorable. 

This paper asks the question: if we add a single molecule to a pre-made solution, how much can that change the solution's thermodynamic equilibrium? 
To make the question quantitative, we define a notion of distance between thermodynamic equilibria,
and we consider scaling with respect to meaningful complexity parameters. 
First, we require an upper limit on the size of molecules in the solution and the analyte, as adding a single very large molecule can trivially affect the entire solution. Large molecules are also expensive to synthesize, and for natural signal detection the structure of the analyte is not under our control. 
Second, we require an upper limit on how many different types of molecules are in the solution, as it is expensive to synthesize new molecular species (though synthesizing many copies is more straightforward).

Our main result is the existence of a family of TBNs that amplify signal exponentially.
In these TBNs, there are exponentially many free ``reporter'' monomers compared to the number of types of monomers and size of monomers.
In the absence of the analyte,
this TBN has a unique stable configuration in which all reporter monomers are bound.
When a single copy of the analyte is added, the resulting TBN has a unique stable configuration in which all reporter monomers are unbound.
These TBNs are parameterized by two values: the first is the amplification factor, determining how many total reporter molecules are freed.
The second is a value we call the system's ``entropy gap'', which determines how thermodynamically unfavorable a configuration of the system would need to be in order for reporters to be spuriously unbound in the absence of the analyte (false positive) or spuriously bound in its presence (false negative).

We also show a corresponding doubly exponential upper bound on the signal amplification problem in TBNs: that given any TBN, adding a single monomer can cause at most a doubly exponential change in its stable configurations. 
We leave as an open question to close this gap:
either proving an exponential upper bound, or giving a TBN with a doubly-exponential amplification factor.

Our work can be compared to prior work on signal amplification that exhibits kinetic barriers.
For example, in reference \cite{minev2021robust}, a detected analyte serves as a seed initiating self-assembly of an arbitrarily long linear polymer. 
In the absence of the analyte, an unlikely kinetic pathway is required for spurious nucleation of the polymer to occur.
However, in that system, false positive configurations are still thermodynamically favorable; if a critical nucleus is able to overcome the kinetic barrier and assemble, then growth of the infinite polymer is equally favorable as from the analyte.
In contrast, in our system, there are \emph{no} kinetic paths, however unlikely, that lead to an undesired yet thermodynamically favored configuration.

\begin{figure}
    \centering
    \includegraphics[width=0.5\columnwidth]{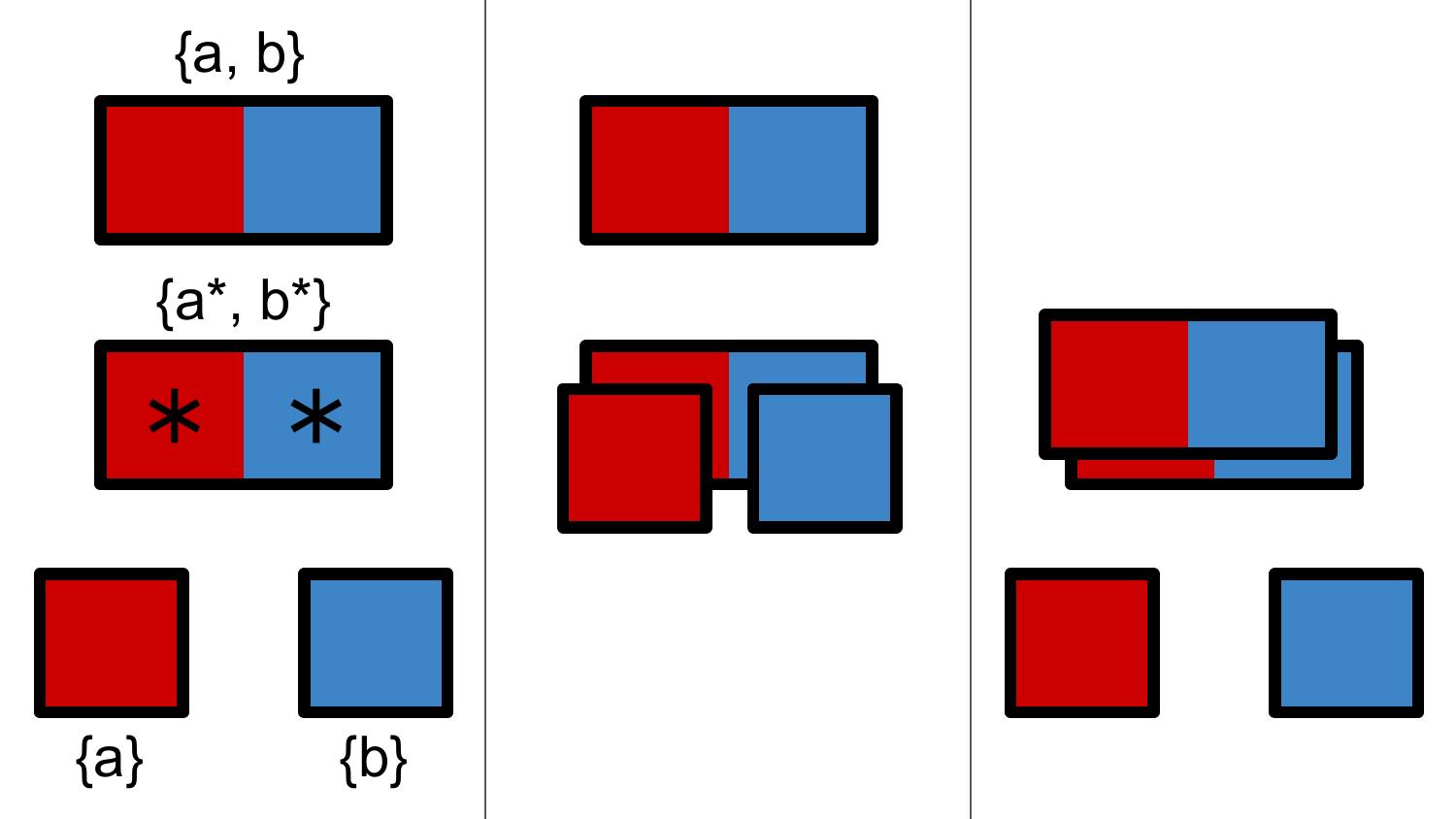}
    \caption{A simple thermodynamic binding network $T$ with four monomers. Site types are differentiated by color. Bonds are shown by juxtaposing monomers so that unstarred sites cover starred sites. Left: the all-singleton configuration $\melt(T)$ with four polymers. Monomers are labeled by their formal identities for reference. Middle: a configuration with two polymers. As all starred sites are covered, the configuration is saturated. Right: a configuration with three polymers. As this has the most possible polymers for a saturated configuration, it is stable.}
    \label{fig:simple_tbn}
\end{figure}

\section{Definitions}
\label{sec:definitions}
\subsection{General TBN Definitions}

A \emph{site type} is a formal symbol such as $a$, and has a \emph{complementary} type, denoted $a^*$, with the interpretation that $a$ binds to $a^*$ (e.g., they could represent complementary DNA sequences). 
We also refer to site types as domain types, and sites as domains. 
We call a site type such as $a$ an \emph{unstarred} site type, and $a^*$ a \emph{starred} site type. 
A \emph{monomer type} is a multiset of site types (e.g., a DNA strand consisting of several binding domains); for example monomer type $\vec{m} = \{a,a,a,b,c^*\}$ has three copies of site $a$, one of site $b$, and one of site $c^*$.
A \emph{TBN}~\cite{doty2017tbn,breik2021kinetic} is a multiset of monomer types. 
We call an instance of a monomer type a \emph{monomer} and an instance of a site type a \emph{site}. 

We take a convention that, unless otherwise specified, TBNs are \emph{star-limiting}: for each site type, there are always at least as many sites of the unstarred type as the starred type among all monomers. Given a TBN, this can always be enforced by renaming site types to swap unstarred and starred types, which simplifies many of the definitions below.

A \emph{configuration} of a TBN is a partition of its monomers into submultisets called \emph{polymers}.
We say that a site type (or a site) on a polymer is \emph{uncovered} if, among the monomers in that polymer, there are more copies of the starred version of that site type than the unstarred version (otherwise \emph{covered}).
A polymer is \emph{self-saturated} if it has no uncovered site types.
A configuration is \emph{saturated} if all its polymers are self-saturated.  A configuration $\alpha$ of a TBN $T$ is \emph{stable} if it is saturated, and no saturated configuration of $T$ has more polymers than $\alpha$. \cref{fig:simple_tbn} shows an example TBN. 

An equivalent characterization of stability is in terms of merges rather than polymer counts. We say that a \emph{merge} is the process of taking two polymers in a configuration and making a new configuration by joining them into one polymer; likewise a \emph{split} is the process of taking one polymer in a configuration and making a new configuration by splitting it into two polymers. Maximizing the number of polymers in a saturated configuration is equivalent to minimizing the number of merges of two polymers necessary to reach a saturated configuration. To this end, some additional notation:

\begin{defn}
    The \emph{distance to stability} of a saturated configuration $\sigma$ is the number of (splits minus merges) necessary to get from $\sigma$ to a stable configuration.
\end{defn} 
Note that this number will be the same for any path of splits and merges, as all stable configurations have the same number of polymers. 

Equivalently, distance to stability is the number of polymers in a stable configuration minus the number of polymers in $\sigma$. We only consider this value for saturated configurations to ensure it is positive and because we may interpret it as a measure of how unlikely we are to observe the network in a given state under the assumption that enthalpy matters infinitely more than entropy. 

The following definitions are not restricted to saturated configurations.

\begin{defn}
    Given a TBN $T$, we say that the all-melted configuration, denoted $\melt(T)$, is the configuration in which all monomers are separate.
\end{defn} 

\begin{defn}
    Given a configuration $\alpha$ in a TBN $T$, its \emph{merginess} $\merginess(\alpha)$ is the number of merges required to get from $\melt(T)$ to $\alpha$ (or equivalently, the number of monomers in $T$ minus the number of polymers in $\alpha$).
\end{defn} 

\begin{defn}
    Given a configuration $\alpha$ in a TBN $T$, its \emph{starriness} $\starriness(\alpha)$ is the number of polymers in $\alpha$ which contain at least one uncovered starred site.
\end{defn} 

We observe that $\alpha$ is saturated if and only if $\starriness(\alpha) = 0$. 

\begin{defn}
    Given configurations $\alpha$ and $\beta$ in a TBN $T$, we say $\alpha \splitsto \beta$ (equivalently, $\beta \mergesto \alpha$) if it is possible to reach $\beta$ from $\alpha$ solely by splitting polymers zero or more times. 
\end{defn} 

We read $\alpha \splitsto \beta$ as ``$\alpha$ splits to $\beta$''. Observe that if $\alpha \mergesto \beta$, then we can reach $\beta$ from $\alpha$ in exactly $\merginess(\beta) - \merginess(\alpha)$ merges. In general, we may order the merges required to go from one configuration to another in whatever way allows the easiest analysis.

\subsection{Comparing TBNs}

We need some notion of how ``different'' two TBNs are, so that we can quantify how much a TBN changes after adding a single monomer.

\begin{defn}[distance between configurations]
Let $\alpha$ and $\beta$ be two configurations of a TBN, or of two TBNs using the same monomer types. We say that the \emph{distance} $d(\alpha,\beta)$ between them is the $L^1$ distance between the vectors of their polymer counts. That is, it is the sum over all types of polymers of the difference between how many copies of that polymer are in $\alpha$ and in $\beta$.
\end{defn}

\begin{defn}[distance between TBNs]
\label{def:distance}
Given TBNs $T$ and $T'$, let $\mathcal C$ and $\mathcal C'$ be their stable configurations. Define the distance between $T$ and $T'$ as 
\begin{equation}
    d(T, T') = \min_{\alpha \in \mathcal C, \alpha' \in \mathcal C'} d(\alpha, \alpha').
\end{equation}
\end{defn}

Note that this distance is not a metric.\footnote{
    In particular, it fails to satisfy the triangle inequality, 
    since $T$ could have a stable configuration close to one of $T'$, so $d(T,T')=1$,
    and $T'$ could have a different stable configuration close to one of $T''$, so $d(T',T'')=1$,
    but $T$ and $T''$ could have no close stable configurations, so $d(T,T'') > 2$.
}
Rather, it is a way to capture how easily we can distinguish between two TBNs; even the closest stable configurations of $T$ and $T'$ have distance $d(T, T')$, so we should be able to distinguish any stable configuration of one of them from all stable configurations of the other by that amount. 

Note that this condition does not directly imply a stronger ``experimentally verifiable'' notion of distance, namely that there is some ``reporter'' monomer which is always bound in one TBN and always free in the other. However, the system we exhibit in this paper does also satisfy this stronger condition. We focus on the distance given here, as it is more theoretically general and our upper bound result in \cref{sec:upperlimit} apply to it.

We also need a notion of how likely we are to observe a configuration of a TBN that is not stable, in order to have a notion of the system being robust to random noise. If a TBN has one stable configuration but many other configurations that are nearly stable, we would expect to observe it in those configurations frequently, meaning that in practice we may not be able to discern what the stable configuration is as easily. 

We work under the assumption that enthalpy matters infinitely more than entropy, so that we may assume that only saturated configurations need to be considered. This assumption is typical for the TBN model, and can be accomplished practically by designing binding sites to be sufficiently strong. Under this paradigm, a configuration's distance to stability is a measure of how unlikely we are to observe it. This motivates the following definition:

\begin{defn}[entropy gap]
    \label{def:entropygap}
    Given a TBN $T$, we say that it has an \emph{entropy gap} of $k$ if, for any saturated configuration $\alpha$ of $T$, one of the following is true:
    \begin{enumerate}
        
        \item $\alpha$ is stable.
        \item There exists some stable configuration $\beta$ such that $\alpha \splitsto \beta$. 
        \item $\alpha$ has distance to stability at least $k$. 
    \end{enumerate}
\end{defn}

Note that by this definition, all TBNs trivially have an entropy gap of one. Note as well that stable configurations are technically also included in the second condition by choosing $\beta = \alpha$, but we list them separately for emphasis.

The second condition is necessary in this definition because any TBN necessarily has some configurations that have distance to stability one, simply by taking a stable configuration and arbitrarily merging two polymers together. These configurations are unavoidable but are not likely to be problematic in a practical implementation, because a polymer in such a configuration should be able to naturally split itself without needing to interact with anything else---these configurations will never be local energy minima. Reference \cite{breik2021kinetic} discusses self-stabilizing TBNs in which \emph{all} saturated configurations have this property, equivalent to an entropy gap of $\infty$.

\subsection{Feed-Forward TBNs}
\begin{defn}
\label{def:feedforward}
We say that a configuration $\alpha$ of a TBN is \emph{feed-forward} if there is an ordering of its polymers such that for each domain type, all polymers with an excess of unstarred instances of that domain type occur before all polymers with an excess of starred instances of that domain type.

We say that a TBN $T$ is \emph{feed-forward} if there is an ordering of its monomer types with this same property---that is, $T$ is feed-forward if $\melt(T)$ is feed-forward.
\end{defn}

For example, the TBN $\{(ab), (a^* c), (b^* c^*)\}$ is feed-forward with this ordering of monomers because the $a, b$ and $c$ come strictly before the $a^*, b^*$ and $c^*$ respectively. Note that not all configurations of a feed-forward TBN are necessarily feed-forward; for instance, merging the first and third monomers in this TBN gives a non-feed-forward configuration.

An equivalent characterization can be obtained by defining a directed graph on the polymers of a configuration $\alpha$ and drawing an edge between any two polymers that can bind to each other, from the polymer with an excess unstarred binding site to the polymer with a matching excess starred binding site (or both directions if both are possible). The configuration $\alpha$ is feed-forward if and only if this graph is acyclic, and the ordering of polymers can be obtained by taking a topological ordering of its vertices.

The main benefit of considering feed-forward TBNs is that we can establish a strong lower bound on the merginess of stable configurations. If \emph{any} TBN $T$ has $n$ monomers that have starred sites, it will always take at least $\frac n 2$ merges to cover all those sites, because each monomer must be involved in at least one merge and any merge can at most bring a pair of them together. For instance, the non-feed-forward TBN $\{\{a, b^*\}, \{a^*, b\}\}$ can be stabilized with a single merge. In feed-forward TBNs, this bound is even stronger, as there is no way to ``make progress'' on covering the starred sites of two different monomers at the same time.

\todo{come up with an elegant/combinatorial way to explain why all the merges in the path of merges for the unique configuration of $T$ are necessary. Is there a condition under which we can say that all the merges in the path must be mandatory, therefore the fact that this is the unique stable configuration follows immediately?}

\begin{lem}
\label{lem:one-merge-per-monomer}
If a configuration $\alpha$ is feed-forward, then any saturated configuration $\sigma$ such that $\alpha \mergesto \sigma$ satisfies $\merginess(\sigma) - \merginess(\alpha) \geq \starriness(\alpha)$. That is, reaching $\sigma$ from $\alpha$ requires at least $\starriness(\alpha)$ additional merges.
\end{lem}

Intuitively, in a feed-forward configuration, the best we can possibly do is to cover all of the starred sites on one polymer at a time. We can never do better than this with a merge like merging $\{a, b^*\}$ and $\{a^*, b\}$ that would let two polymers cover all of each others' starred sites.

\begin{proof}
Given a feed-forward configuration $\alpha$, let $L$ be the ordered list of polymers from $\alpha$ being feed-forward. Partition $L$ into separate lists (keeping the ordering from $L$) based on which polymers are merged together in $\sigma$. That is, for each fully merged polymer $\mathbf{P} \in \sigma$ create a list $L_\mathbf{P}$ of the polymers from $\alpha$ that are merged to form $\mathbf{P}$, and order this list based on the ordering from $L$. We can order the merges to reach $\sigma$ from $\alpha$ as follows: repeatedly (arbitrarily) pick a polymer $\mathbf{P}$ from $\sigma$ and merge all of the polymers in $L_{\mathbf{P}}$ together in order (merge the first two polymers in $L_{\mathbf{P}}$, then merge the third with the resulting polymer, and so on).\\\\
This sequence of merges gives us a sequence of configurations $\alpha = \alpha_1, \alpha_2, \ldots, \alpha_\ell = \sigma$. We observe that for $1 \leq i \leq \ell - 1$, we have $\starriness(\alpha_i) - \starriness(\alpha_{i+1}) \leq 1$. That is, each merge can lower the starriness by at most one. We know this because each merge is merging a polymer $\mathbf{Q} \in \alpha$ with one or more other already-merged polymers from $\alpha$ that all come before $\mathbf{Q}$ in $L$. This means $\mathbf{Q}$ cannot cover any starred sites on any monomers it is merging with. The only way for the starriness of a configuration to decrease by more than 1 in a single merge is for the two merging polymers to cover all of each others' starred sites, so it follows that each merge in this sequence lowers starriness by at most 1. From this it follows that we need at least $\starriness(\alpha)$ merges to get to $\sigma$, because $\starriness(\sigma) = 0$.

\end{proof}

Letting $\alpha = \melt(T)$ (note $m(\alpha)=0$) gives the following corollary.

\begin{cor}
\label{cor:one-merge-per-monomer}
Any saturated configuration $\sigma$ of a feed-forward TBN $T$ satisfies $\merginess(\sigma) \geq \starriness(\melt(T))$.
\end{cor}

Because stable configurations are saturated configurations with the minimum possible merginess, this bound gives the following corollary.

\begin{cor}
\label{cor:stable-if-requires-exactly-s-merges}
If a saturated configuration $\sigma$ of a feed-forward TBN $T$ satisfies $\merginess(\sigma) = \starriness(\melt(T))$,
then $\sigma$ is stable.
\end{cor}

\section{Signal Amplification TBN}
\subsection{Amplification Process}
\label{subsec:amplification}
In this section, we prove our main theorem. This theorem shows the existence of a TBN parameterized by two values $n$ (the amplification factor) and $k$ (the entropy gap). Intuitively, this TBN amplifies the signal of a single monomer by a factor of $2^n$, with any configurations that give ``incorrect'' readings having $\Omega(k)$ distance to stability. Our proof will be constructive.
\begin{thm}
\label{thm:mainthm}For any integers $n \geq 1, k \geq 2$, there exists a TBN $T = T_{n,k}$ and monomer $\mathbf{a}$ (the analyte) such that if $T^{\mathrm{a}} = T_{n,k}^\mathrm{a}$ is the TBN obtained by adding one copy of $\mathbf{a}$ to $T_{n,k}$, then
\begin{enumerate}
    \item $T$ and $T^{\mathrm{a}}$ each have exactly one stable configuration, denoted $\sigma_{n,k}$ and $\sigma^{\mathrm{a}}_{n,k}$ respectively, with $d(\sigma_{n,k}, \sigma^{\mathrm{a}}_{n,k}) \geq 2^n$. In particular, there are $k$ monomer types with $2^{n-1}$ copies each, with all of these monomers bound in $\sigma_{n,k}$ and unbound in $\sigma_{n,k}^{\mathrm{a}}$.
    \item $T$ and $T^{\mathrm{a}}$ each have an entropy gap of $\lfloor \frac k 2 \rfloor - 1$.
    \item $T$ and $T^{\mathrm{a}}$ each use $\mathcal O(nk)$ total monomer types, $\mathcal O(nk^2)$ domain types, and $\mathcal O(k^2)$ domains per monomer.
\end{enumerate} \end{thm}
The first condition implies that $T_{n,k}$ can detect a single copy of $\mathbf{a}$ with programmable exponential strength - there is only one stable configuration either with or without $\mathbf{a}$, and they can be distinguished by an exponential number of distinct polymers. Note that this is even stronger than saying that $d(T_{n,k}, T_{n,k}^\mathrm{a}) \geq 2^n$, as that statement would allow each TBN to have multiple stable configurations. The second condition implies that the system has a programmable resilience to having incorrect output, because configurations other than the unique stable ones in each case are ``programmably'' unstable (based on $k$), and thus programmably unlikely to be observed. Note that throughout this paper we will  use $\frac k 2$ instead of $\lfloor \frac k 2 \rfloor$ for simplicity, as we are concerned mainly with asymptotic behavior. The third condition establishes that the system doesn't ``cheat'' - it doesn't obtain this amplification by either having an extremely large number of distinct monomers, or by having any single large monomers.

\begin{figure}
        \centering 
        \includegraphics[width=0.95\columnwidth]{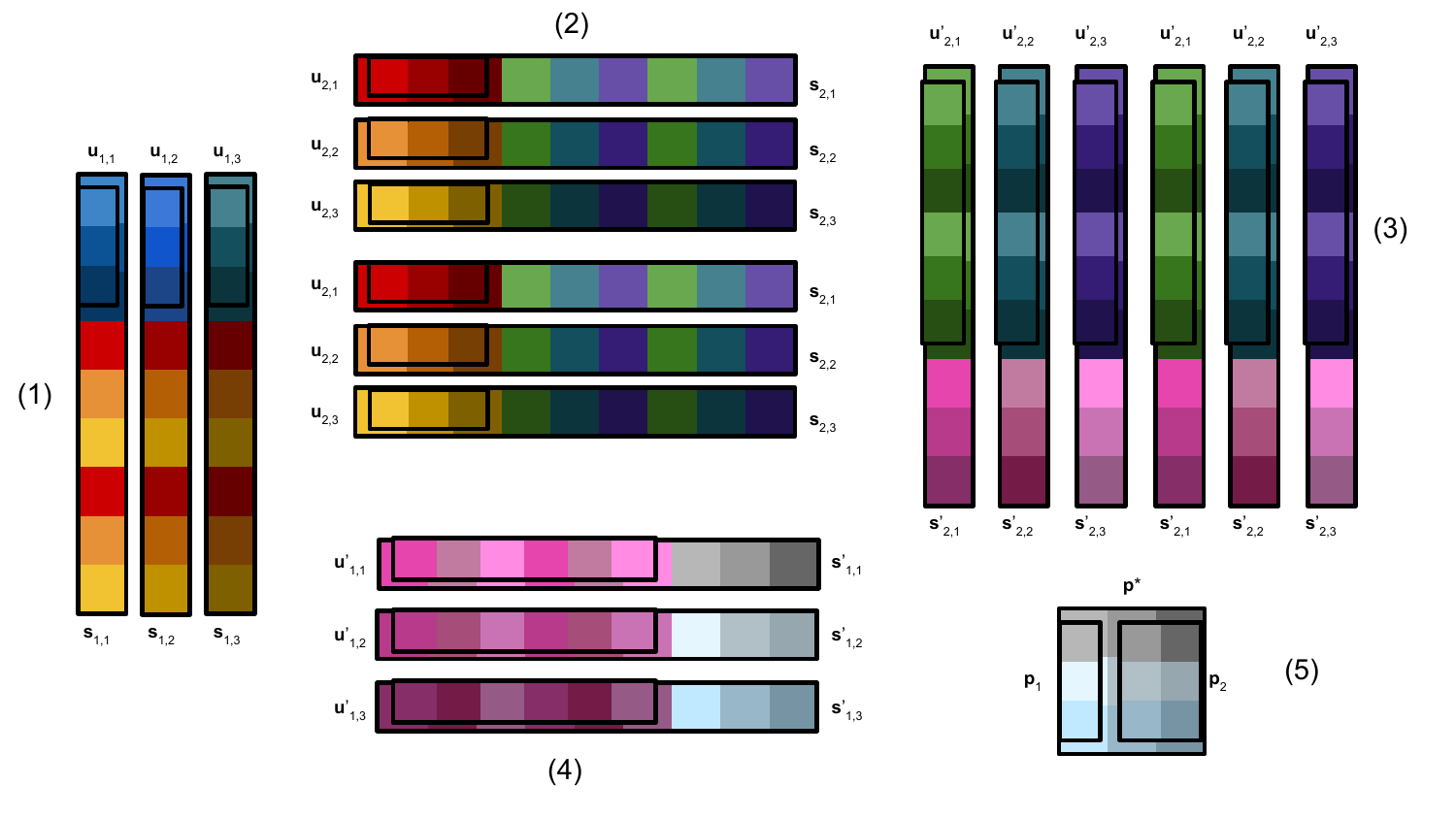}
        \caption{The unique stable configuration $\sigma_{2,3}$ of $T_{2,3}$, with 19 polymers. All starred sites are visually ``covered'' by unstarred sites on another monomer. The parts of the diagram are numbered by the order that the signal from the analyte will cascade through them. Parts (1) and (2) form the ``first half'', where the signal is doubled at each step. Parts (3) and (4) form the ``second half'', where the signal converges so that it can get an ``entropic payoff'' from part (5).}
        \label{fig:pre_analyte}
        
        \centering
        \includegraphics[width=0.95\columnwidth]{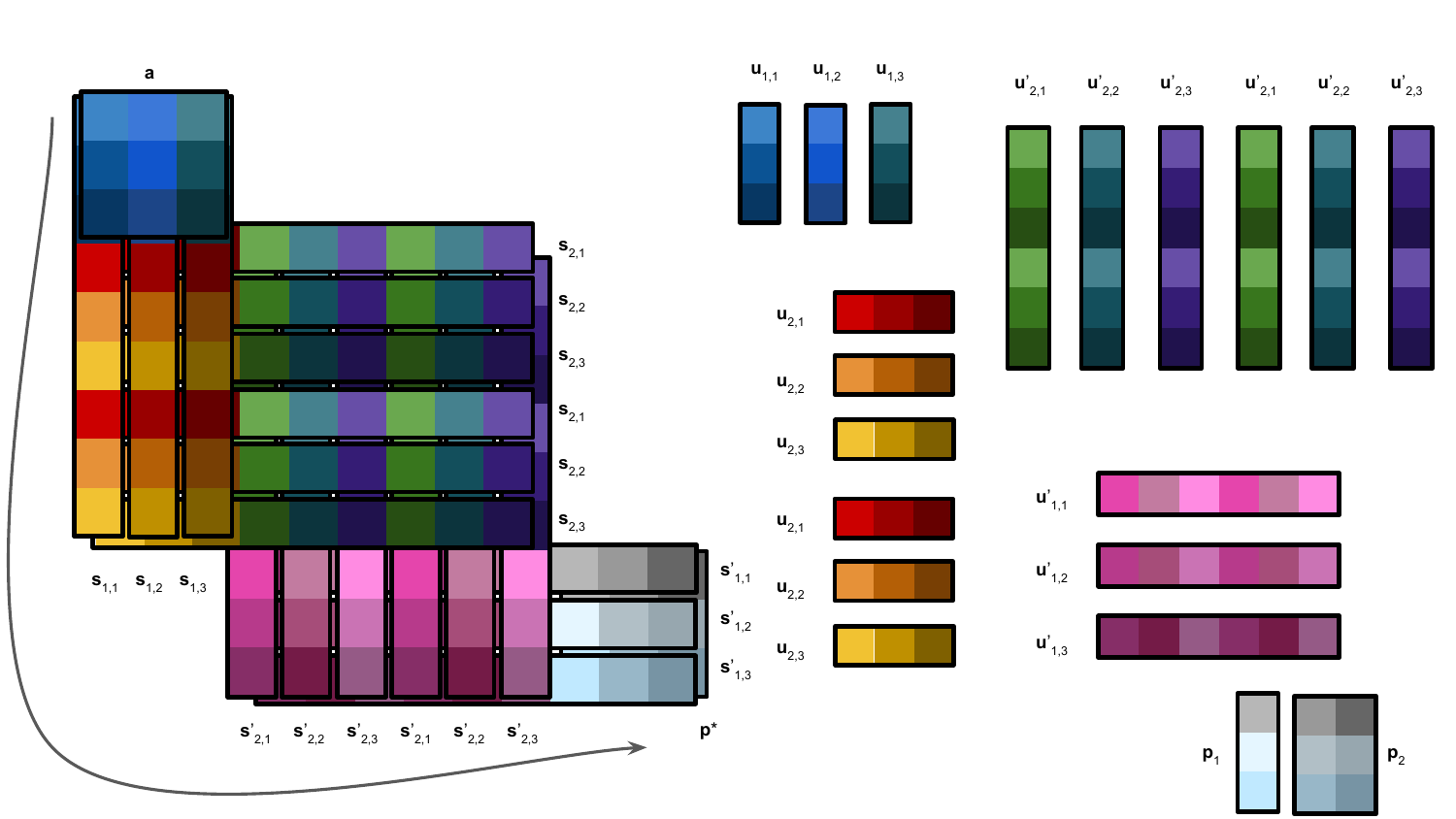}
        \caption{The unique stable configuration $\sigma^{\mathrm{a}}_{2,3}$ of $T^{\mathrm{a}}_{2,3}$. The arrow shows the conceptual order in which the analyte's signal has been propagated, with $\mathbf{a}$ covering all $\mathbf{s}_{1,j}$, which cover all $\mathbf{s}_{2,j}$, which cover all $\mathbf{s}'_{2,j}$, which cover all $\mathbf{s}'_{1,j}$, which finally cover $\mathbf{p}^*$. This configuration has 21 polymers, 2 more than $\sigma_{2,3}$: conceptually, one of these is from adding the analyte and the other is from the analyte's signal cascading through the layers to release $P_1$ and $P_2$ at the cost of one merge. As they have no polymers in common, $d(\sigma_{2,3}, \sigma^{\mathrm{a}}_{2,3}) = 19 + 21 = 40$.}
        \label{fig:with_analyte}
\end{figure}

\begin{toappendix}
\begin{figure}
    \centering
    \includegraphics[width=0.95\columnwidth]{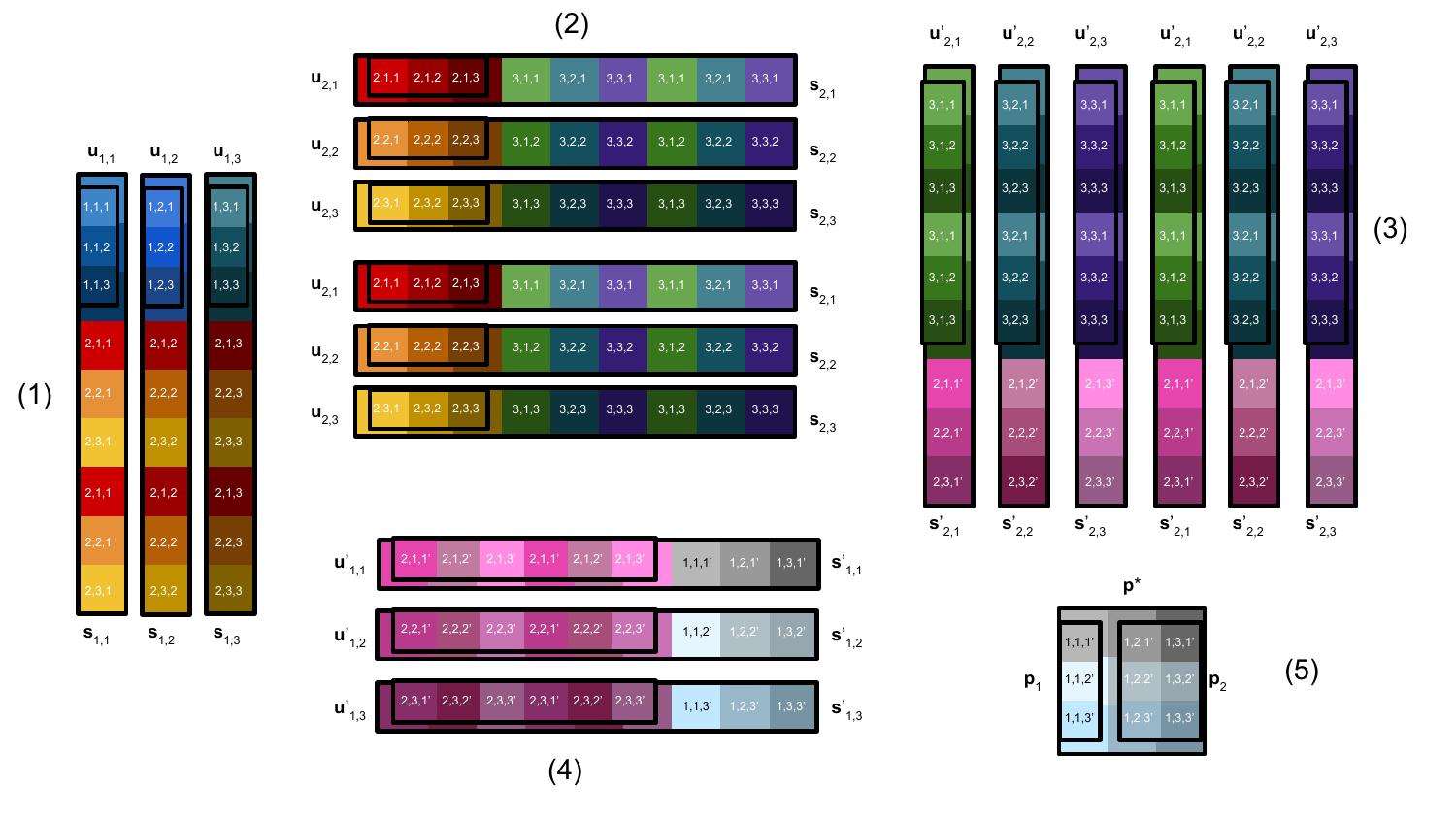}
    \caption{A version of \cref{fig:pre_analyte} with text labels on domains, for accessibility and allowing comparison with the domains as defined in the text of the paper. This figure shows the unique stable configuration of $T_{2,3}$.}
    \label{fig:text_labels}

    \centering
    \includegraphics[width=0.95\columnwidth]{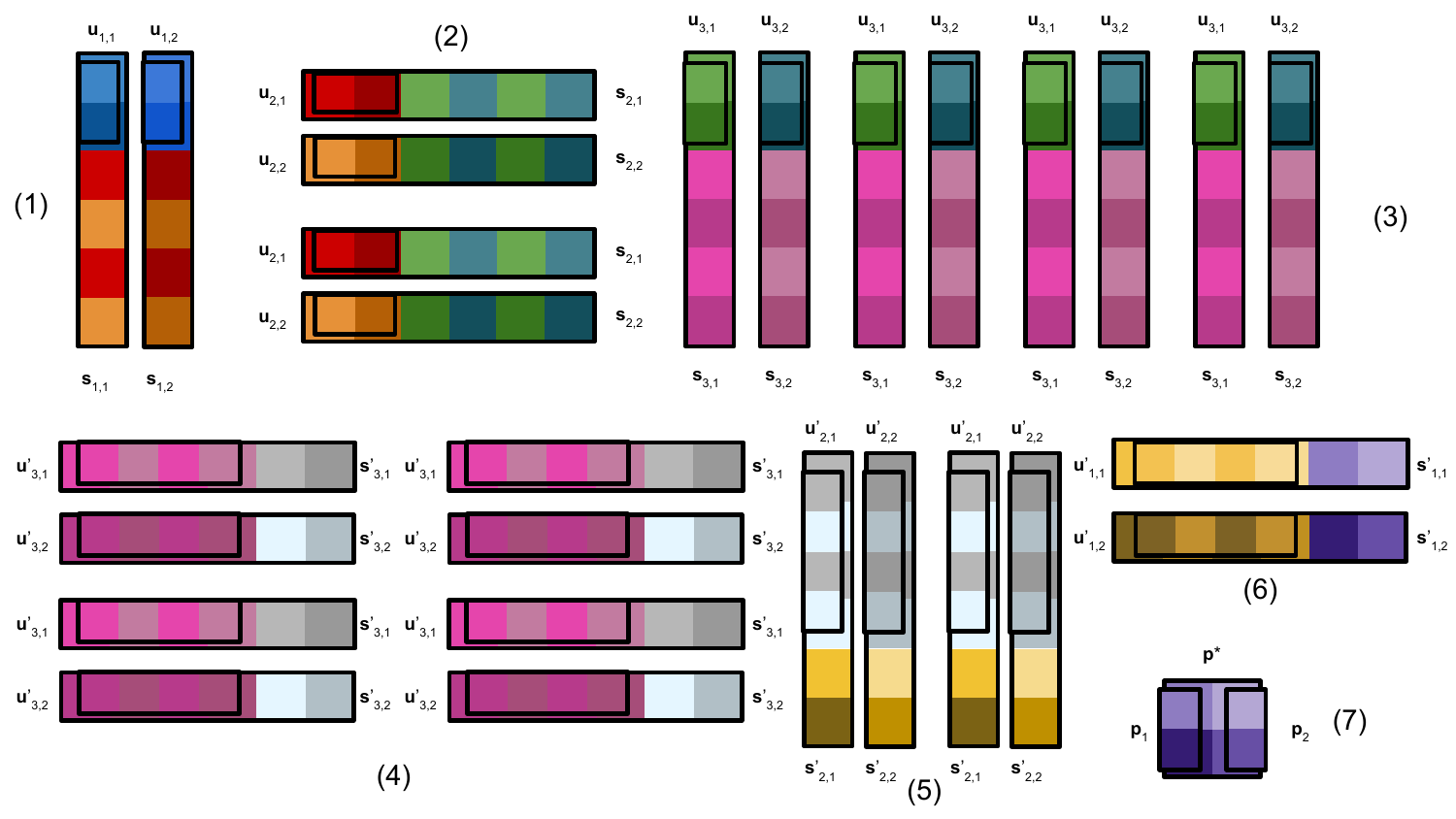}
    \caption{The unique stable configuration $\sigma_{3,2}$ of $T_{3,2}$. Compared to \cref{fig:pre_analyte} (or \cref{fig:text_labels} above), which show $T_{2,3}$, this figure shows one more layer and a smaller entropy gap parameter. The additional layer means that there are 4 copies of each monomer in the largest parts of the figure, compared to 2 copies of each monomer in the other figures; if another layer were added, it would contain 8 copies of each monomer. The smaller entropy gap parameter manifests in this figure visually having a ``2 by 2 grid'' design motif, compared to the ``3 by 3 grid'' motif in the other figures.}
    \label{fig:more_layers}
    
\end{figure}
\end{toappendix}

The entire TBN $T_{n,k}$ is depicted in Figures \ref{fig:pre_analyte} and \ref{fig:with_analyte} with $n = 2$ and $k = 3$. The former shows the unique stable configuration before adding the analyte, and the latter shows the unique stable configuration after adding the analyte. For comparison, \cref{fig:more_layers} depicting the pre-analyte configuration with $n = 3$ and $k = 2$ is shown in the appendix.

We will start by constructing the first half of $T_{n,k}$ and describing ``how it works''. The monomers in this first half are the driving force that allows $T_{n,k}$ and $T_{n,k}^\mathrm{a}$ to have exponentially different stable configurations.

The first half of $T_{n,k}$ has monomer types $\mathbf{u}_{i,j}$ and $\mathbf{s}_{i,j}$ (named as those with only unstarred sites, and those with both starred and unstarred sites) for $1\leq i \leq n, 1 \leq j \leq k$. It has domain types denoted as triples $(i,j,\ell)$ for $1 \leq i \leq n + 1, 1 \leq j, \ell \leq k$. Each $\mathbf{u}_{i,j}$ monomer has $k$ different unstarred domains, one of each $(i,j,\ell)$ for each $1 \leq \ell \leq k$. Each $\mathbf{s}_{i,j}$ monomer has a starred copy of each domain in $\mathbf{u}_{i,j}$, and additionally has two copies of each unstarred domain $(i+1,\ell,j)$ for each $1 \leq \ell \leq k$ (note that here the second domain type parameter varies instead of the third). For each $\mathbf{u}_{i,j}$ and $\mathbf{s}_{i,j}$ monomer, there are $2^{i-1}$ copies. We can conceptually break these monomers into $n$ ``layers'', each consisting of all monomers with the same value for their first parameter. The analyte we wish to detect, $\mathbf{a}$, is a monomer that has one copy of each unstarred domain $(1,j,\ell)$, $1 \leq j, \ell \leq k$. 

Conceptually, when the analyte is absent, the most efficient way for all starred sites on each $\mathbf{s}_{i,j}$ to be covered is by the unstarred sites on a corresponding $\mathbf{u}_{i,j}$, as seen in \cref{fig:pre_analyte}. Although the TBN model is purely thermodynamic, we can conceptualize that when the analyte is added, its signal can propagate ``kinetically'' through each layer. In the first layer, it can ``displace'' the $k$ different $\mathbf{u}_{1,j}$ monomers and bind to all of the $\mathbf{s}_{1,j}$ monomers. In doing so, it brings together all the unstarred sites on all of the $\mathbf{s}_{1,j}$ monomers. Having been brought together, these sites ``look like'' two copies of the analyte, but with the domains from layer 2 instead of layer 1. Thus, this polymer is then able to displace \emph{two} copies of each $\mathbf{u}_{2,j}$ from their corresponding $\mathbf{s}_{2,j}$ monomers, thus bringing all of the $\mathbf{s}_{2,j}$ together. This in turn now looks like four copies of the analyte for the domains in the third layer, and so on. Each layer allows this polymer to assimilate exponentially more $\mathbf{s}_{i,j}$, thus freeing exponentially many $\mathbf{u}_{i,j}$. Each of these displacement steps involves an equal number of splits and merges.

\subsection{Convergence Process}
\label{subsec:convergence}
So far, the TBN described has exactly one stable configuration before adding the analyte, and it performs the task of amplifying signal by having the potential to change its state exponentially when the analyte is added. However, there is also a stable configuration after adding the analyte in which nothing else changes, and many others in which only a small amount of change occurs. We must guaranteed that the analyte's signal ``propagates'' through all of the layers.

To design the system to meet this requirement, we observe that all exponentially many monomers that have been brought together must contribute to some singular change in the system that gains some entropy, to spur the signal into propagating. The typical way to accomplish this in a TBN is by having monomers that have been brought together displace a larger number of monomers from some complex at the cost of a smaller number of merges. Because the pre-analyte TBN has an entropy gap of $k - 1$ in this design so far, we can afford to give the TBN with the analyte an ``entropic payoff'' of $\frac k 2$. When the analyte is absent, this payoff is weak enough that there will still be an entropy gap of $\frac k 2 - 1$; when the analyte is present, the existence of this payoff will force the signal to fully propagate, and will give the TBN with the analyte an entropy gap of $\frac k 2  - 1$ by making it so that any configurations in which this payoff is not achieved are also far away from stable. 

Another challenge is that we cannot simply detect all our exponentially many conjoined monomers by binding them all to a single exponentially large monomer, because we need to bound the size of the largest monomer in the system. Our conceptual strategy for overcoming this is as follows: the signal will converge in much the same way as it was amplified. In the amplification step, one set of domains coming together in one layer was enough to cause two of them to come together in the next layer. In this convergence step, two sets of domains in one layer will have to converge together to activate one set in the next layer. This convergence ends in bringing together a set of binding sites that is of the same size as the analyte, which can then directly displace some monomers to gain $\frac k 2$ total polymers.

We now fully define $T_{n,k}$. We start with the already described $\mathbf{u}_{i,j}$ and $\mathbf{s}_{i,j}$. To these, we first add monomer types $\mathbf{u}'_{i,j}$ and $\mathbf{s}'_{i,j}$ for $1 \leq i \leq n, 1 \leq j \leq k$. These monomers are the `converging' equivalents of $\mathbf{u}_{i,j}$ and $\mathbf{s}_{i,j}$. Conceptually, they will activate in the reverse order: two copies of each $\mathbf{s}'_{i,j}$ for $1 \leq j \leq k$, when brought together, will be able to bring together one copy of each $\mathbf{s}'_{i-1,j}$ for $1 \leq j \leq k$.

Each $\mathbf{u}'_{i,j}$ monomer has $2k$ unstarred domains: two copies each of domains $(i+1,j,\ell)'$ for each $1 \leq \ell \leq k$. Each $\mathbf{s}'_{i,j}$ has a starred copy of each of the $2k$ domains in $\mathbf{u}'_{i,j}$ and additionally has one unstarred domain $(i,\ell,j)'$ for $1 \leq \ell \leq k$ (note again here that the second domain type parameter varies instead of the third). One exception is the monomers $\mathbf{u}'_{n,j}$ and $\mathbf{s}'_{n,j}$ (the first ones to activate) which use domains $(n+1,j,\ell)$ and $(n+1,\ell,j)$ instead of $(n+1,j,\ell)'$ and $(n+1,\ell,j)'$ respectively so that they can interact with $\mathbf{s}_{n,j}$ monomers that have been brought together. Each monomer $\mathbf{u}'_{i,j}$ and $\mathbf{s}'_{i,j}$ has $2^{i-1}$ copies. 

Finally, we add ``payoff'' monomers that will yield an entropic gain of $\frac k 2$ when the signal from the analyte has cascaded through every layer. This choice of $\frac k 2$ is arbitrary---a similar design works for any integer between 1 and $k$. Choosing a higher value leads to a higher entropy gap after adding the analyte and a lower entropy gap before adding it, and vice versa choosing a lower value. For simplicity of definitions we will assume $k$ is even (though figures are shown with $k = 3$, which shows how to generalize to odd $k$). 

We add one monomer $\mathbf{p}^*$, which contains the $k^2$ sites $(1,\ell_1,\ell_2)'^*$ for $1 \leq \ell_1, \ell_2 \leq k$. Note that this monomer can be replaced with $k$ monomers of size $k$ (in which case $\mathbf{a}$ would be the only monomer with more than $3k$ domains), but doing so makes the proof more complex. The idea is that when all $\mathbf{s}'_{1,j}$ monomers are already together (as they can be ``for free'' when $\mathbf{a}$ is present), they can cover $\mathbf{p}^*$ in one merge; if they are apart, this requires $k$ merges. In order to make this favorable to happen when they're already together but unfavorable when they're initially apart, we add another way to cover $\mathbf{p}^*$ that takes $\frac {k}{2}$ merges. This is accomplished via monomers $\mathbf{p}_j$ for $1 \leq j \leq \frac k 2$. Each $\mathbf{p}_j$ contains the $2k$ sites $(1,2j - 1,\ell)$ and $(1,2j,\ell)$ for $1 \leq \ell \leq k$. We can interpret this geometrically as $\mathbf{p}^*$  being a square, the $\mathbf{s}'_{1,j}$ covering it by rows, and the $\mathbf{p}_j$ covering it by two columns at a time. This completes the definition of $T_{n,k}$. Recall $T_{n,k}^\mathrm{a}$ is $T_{n,k}$ with one added copy of $\mathbf{a}$.

\begin{lemrep}
    \label{lem:pre-analyte-one-stable-config}
    $T_{n,k}$ has exactly one stable configuration $\sigma_{n,k}$.
\end{lemrep}

\begin{proof}
    We consider merges to get from the melted configuration to any saturated configuration. We may order these merges such that we first make all the merges necessary to cover each individual $\mathbf{s}_{i,j}$ in increasing value of $i$, then each individual $\mathbf{s}'_{i,j}$ in decreasing value of $i$. We see that at each step of this process, we may cover the monomer in question by a single merge (of its corresponding $\mathbf{u}_{i,j}$ or $\mathbf{u}'_{i,j}$). If we never merge the corresponding $\mathbf{u}$ monomer, the only other monomers that can cover the starred sites on a given $\mathbf{s}_{i,j}$ are $k$ different $\mathbf{s}_{i-1,\ell}$ monomers. Likewise, the only other way to cover the starred sites on a given $\mathbf{s}'_{i,j}$ is by using $k$ different $\mathbf{s}'_{i+1,\ell}$ monomers (except for $\mathbf{s}'_{n,j}$ which needs $\mathbf{s}_{n,\ell}$ monomers). 

    If an $\mathbf{s}$ or $\mathbf{s}'$ monomer is covered in multiple different ways, we order the merges such that it is first covered by one corresponding $\mathbf{u}$ monomer (and then ignore any other merges for now, as we are still ordering the merges to cover each $\mathbf{s}$ monomer sequentially). We see then that if every $\mathbf{s}$ monomer is covered by a $\mathbf{u}$ monomer, then no $\mathbf{s}$ monomers will be brought together during this process. Therefore, the first time in this sequence that we choose to cover an $\mathbf{s}$ without its corresponding $\mathbf{u}$ will require $k$ total merges to cover that $\mathbf{s}$. The resulting configuration is feed-forward, so by  \cref{cor:one-merge-per-monomer}, reaching a stable configuration requires at least one more merge per remaining $\mathbf{s}$ monomer. This results in at least $k - 1$ extra merges compared to covering $\mathbf{s}$ and $\mathbf{s}'$ monomers by using $\mathbf{u}$ and $\mathbf{u}'$ monomers respectively. 

    Once all $\mathbf{s}$ and $\mathbf{s}'$ monomers are covered, the only other monomer with starred sites is $\mathbf{p}^*$, so we can make all the merges that are needed to cover it. If none of the $\mathbf{s}'_{1,j}$ monomers have been brought together, then the fewest merges it takes to cover $\mathbf{p}^*$ is $\frac k 2$, via the $\mathbf{p}_j$ monomers. If any of them have been brought together, then it could potentially take a single merge to cover $\mathbf{p}^*$. However, this would have required $k - 1$ extra merges at some point during the covering of $\mathbf{s}$ monomers, resulting in $\frac k 2$ extra total merges compared to covering all $\mathbf{s}$ monomers with $\mathbf{u}$ monomers, then covering $\mathbf{p}^*$ with $\mathbf{p}_j$ monomers.

    Therefore, this latter set of merges covers all starred sites in as few merges as possible, and therefore gives the unique stable configuration of $T_{n,k}$. 
    
\end{proof}

\begin{correp}
    \label{cor:pre_analyte_entropy_gap}
    $T_{n,k}$ has an entropy gap of $\frac k 2 - 1$.
\end{correp}

\begin{proof}
    Recall \cref{def:entropygap} for what we must show. Any saturated configuration that does not make all the merges in $\sigma_{n,k}$ must either have some $\mathbf{s}$ that is not covered by its corresponding $\mathbf{u}$ (resulting in at least $\frac k 2$ extra merges, as per the above argument), or must cover $\mathbf{p}^*$ with initially-separate $\mathbf{s}'_{1,j}$ monomers (resulting in $\frac k 2$ extra merges). Thus, any such configuration has distance to stability at least $\frac k 2$. Any other saturated configuration that does make all of the merges in this sequence simply makes some extra merges afterward, and therefore splits to $\sigma_{n,k}$. It follows that $T_{n,k}$ has an entropy gap of $\frac k 2$ (and also of $\frac k 2 - 1$, for consistency in the statement of \cref{thm:mainthm}).
\end{proof}

\begin{lemrep}
    \label{lem:post_analyte_entropy_gap}
    $T_{n,k}^\mathrm{a}$ has exactly one stable configuration $\sigma^{\mathrm{a}}_{n,k}$, and $T_{n,k}^\mathrm{a}$ has an entropy gap of $\frac k 2 - 1$.
\end{lemrep}

\begin{proof}
    We see that $T_{n,k}^\mathrm{a}$ (like $T_{n,k}$) is feed-forward (recall \cref{def:feedforward}) by first ordering $\mathbf{a}$ along with all the $\mathbf{u}_{i,j}$, $\mathbf{u}'_{i,j}$, and $\mathbf{p}_j$ monomers (none of which have starred sites), then all the $\mathbf{s}_{i,j}$ in increasing order of $i$, then all the $\mathbf{s}'_{i,j}$ in decreasing order of $i$, and finally $\mathbf{p}^*$. 

    Unlike $T_{n,k}$, however, we may reach a stable state by merging $\mathbf{a}$ together with every single $\mathbf{s}_{i,j}$, every single $\mathbf{s}'_{i,j}$ and $\mathbf{p}^*$ into a single polymer. This covers all starred sites, and requires exactly one merge per monomer with starred sites, so by \cref{cor:stable-if-requires-exactly-s-merges}, this configuration $\sigma^{\mathrm{a}}_{n,k}$ is stable.

    Now, we examine an arbitrary saturated configuration $\sigma$ of $T_{n,k}^\mathrm{a}$. We consider merges in essentially the opposite order of how they were considered when analyzing $T_{n,k}$. First, consider $\mathbf{p}^*$. It must be covered either by all the $\mathbf{p}_j$ monomers, or by all the $\mathbf{s}'_{1,j}$ monomers. If we merge all the $\mathbf{p}_j$ monomers to $\mathbf{p}^*$, we arrive at a configuration that is still feed-forward, but has only one fewer polymer with uncovered starred sites compared to $\melt(T_{n,k}^\mathrm{a})$ in spite of making $\frac k 2$ merges. Therefore, by \cref{lem:one-merge-per-monomer}, reaching a saturated configuration from this point requires at least $\frac k 2 - 1$ extra merges compared to $\sigma^{\mathrm{a}}_{n,k}$.

    Now, we may make a similar argument for all $\mathbf{s}$ monomers in the opposite order that we considered them in 
    \cref{lem:pre-analyte-one-stable-config}. First, either we have already made $\frac k 2 - 1$ extra merges, or the $\mathbf{s}'_{1,j}$ monomers have all been brought together on a single polymer to cover $\mathbf{p}^*$. If we now make all the merges necessary to cover all starred sites on this polymer, we must do so either using all the $\mathbf{u}'_{1,j}$ or by using all the $\mathbf{s}'_{2,j}$. If we use the former, then this will require $k$ total merges but will only reduce the count of polymers with starred sites by 1. The resulting configuration is still feed-forward, so again by \cref{lem:one-merge-per-monomer} any saturated configuration we reach from this point will require at least $k - 1$ extra merges compared to $\sigma$. Otherwise, we must bring all the $\mathbf{s}_{2,j}$ monomers together to cover these sites. This does not fall victim to the same argument, because bringing these monomers with starred sites together onto the same polymer lowers the total number of polymers with uncovered starred sites. Now that they have been brought together, the same argument shows that we must either cover all the starred sites on the $\mathbf{s}'_{2,j}$ using all the $\mathbf{s}'_{3,j}$, or suffer $k - 1$ extra merges. The same argument for each layer in the converging part of the TBN also works for each layer in the amplifying part. Finally, after running through this argument we arrive at all $\mathbf{s}_{1,j}$ being brought together, which can be covered either by a single merge of $\mathbf{a}$ or by merging the $k$ $\mathbf{u}_{1,j}$ to it.

    Overall, this shows that any saturated configuration of $T_{n,k}^\mathrm{a}$ either makes all of the merges in $\sigma^{\mathrm{a}}_{n,k}$ or it must make at least $\frac k 2 - 1$ extra merges. It follows that $\sigma^{\mathrm{a}}_{n,k}$ is the unique stable configuration of $T_{n,k}^\mathrm{a}$, with an entropy gap of $\frac k 2 - 1$ as desired. 
\end{proof}

Proofs of these results are left to the appendix.

These results together complete the proof of \cref{thm:mainthm}: each of the more than $2^n$ $\mathbf{u}$ and $\mathbf{u}'$ monomers (which serve as reporters) are bound in $\sigma_{n,k}$ and unbound in $\sigma^{\mathrm{a}}_{n,k}$, implying their distance is more than $2^n$. The largest monomer is $\mathbf{a}$ with $k^2$ domains, and there are $(2n+1)k^2$ domain types and $4nk$ monomer types for the $\mathbf{s}_{i,j}$, $\mathbf{u}_{i,j}$, $\mathbf{s}'_{i,j}$, and $\mathbf{u}'_{i,j}$, plus $2 + \frac k 2$ more for $\mathbf{a}$, $\mathbf{p}^*$, and $\mathbf{p}_j$.

\subsection{Avoiding Large Polymer Formation}
\label{subsec:translators}

\begin{figure}
    \centering
    \includegraphics[width=0.95\columnwidth]{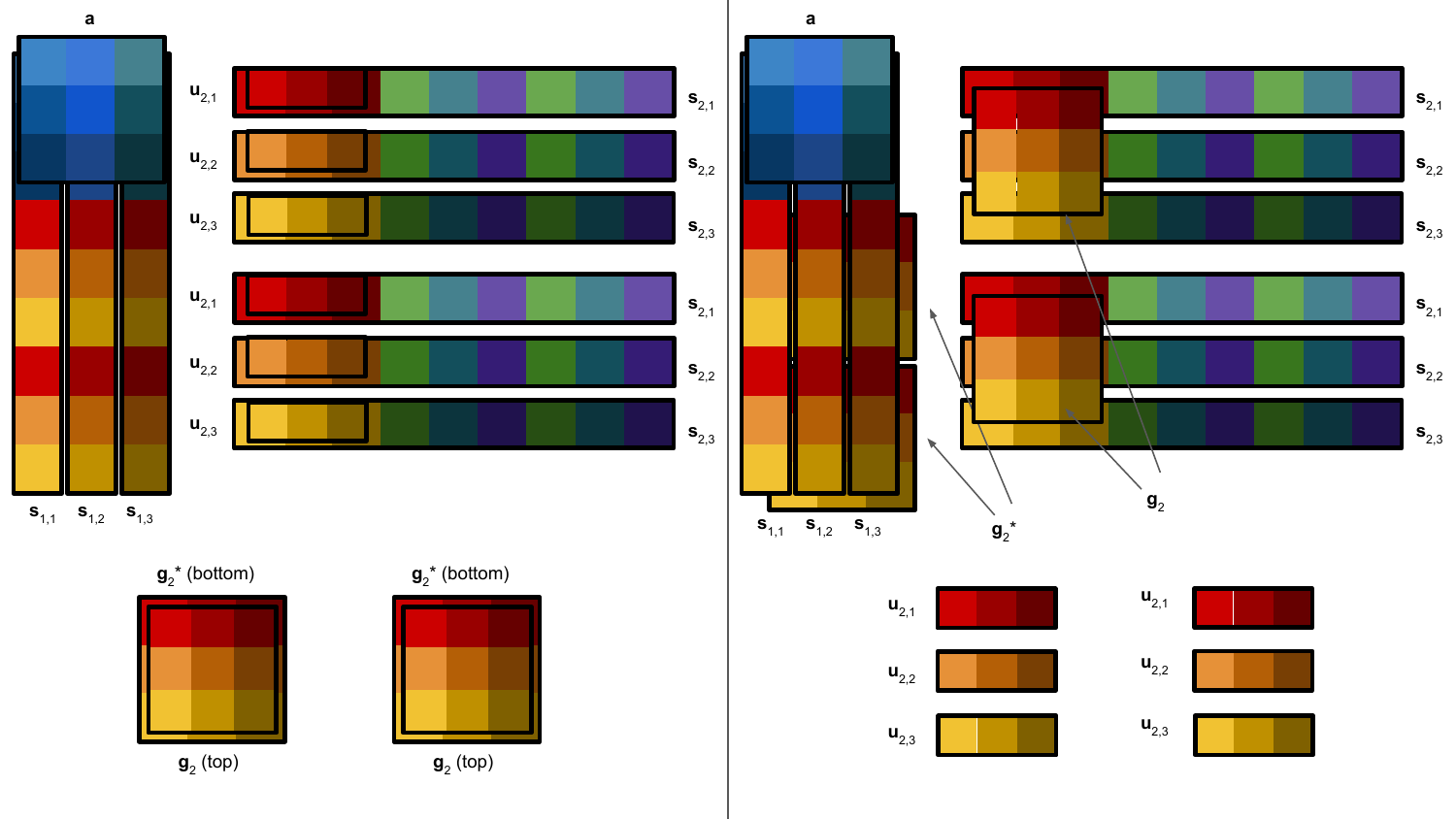}
    \caption{The amplifying translator gadget, before and after it is triggered to propagate the signal forward by one layer. When $\mathbf{s}_{1,1}$, $\mathbf{s}_{1,2}$ and $\mathbf{s}_{1,3}$ have been brought together, instead of directly replacing all the $\mathbf{u}_{2,j}$ monomers, they can split two $\{\mathbf{g}_2, \mathbf{g}_2^*\}$ complexes, and the $\mathbf{g}_2$ monomers can replace the $\mathbf{u}_{2,j}$ monomers.}
    \label{fig:translator_amplifying}
    
    \centering
    \includegraphics[width=0.95\columnwidth]{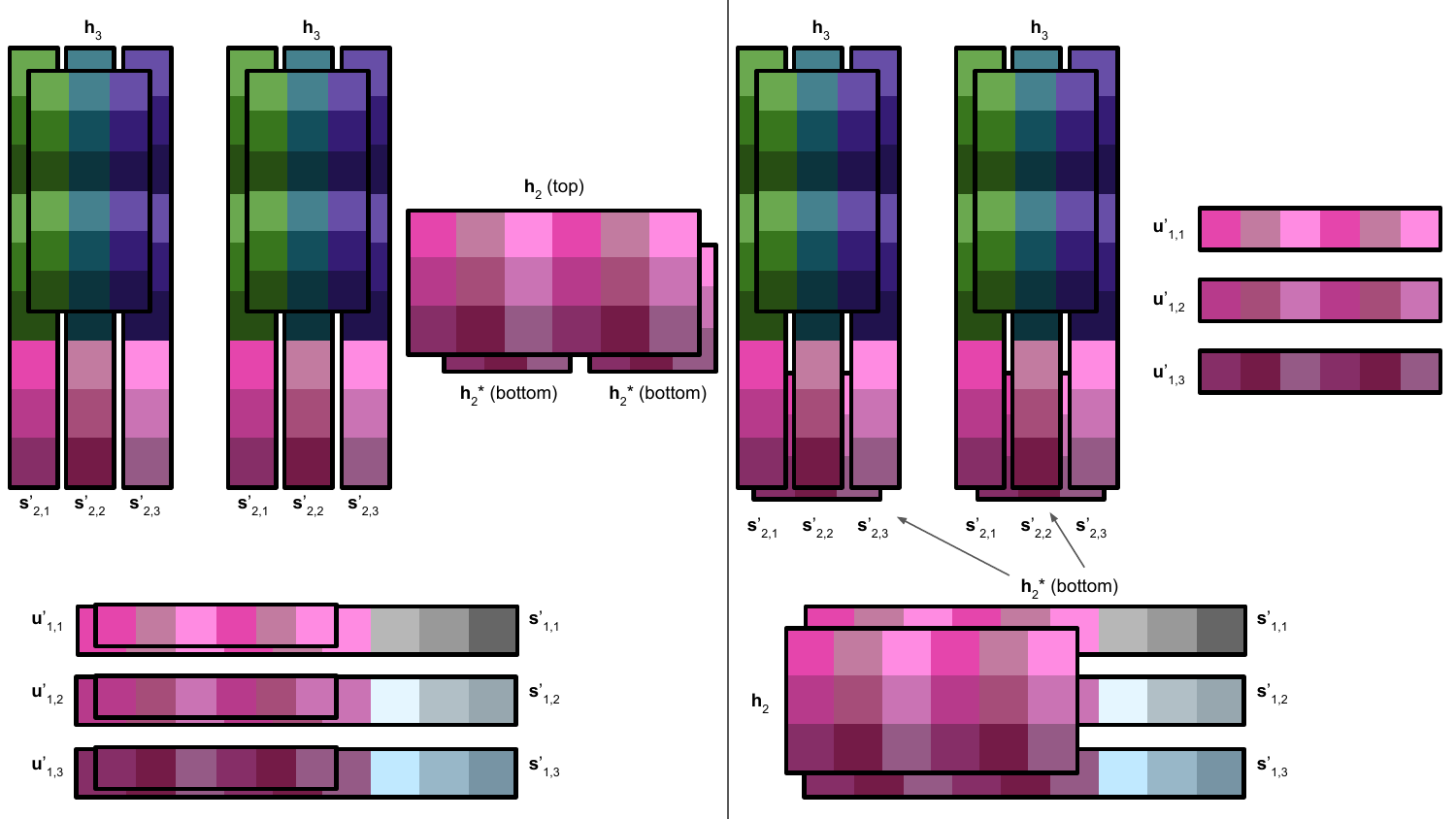}
    \caption{The converging translator gadget, before and after it is triggered to propagate the signal forward by one layer. When two copies each of $\mathbf{s}'_{2,1}$, $\mathbf{s}'_{2,2}$ and $\mathbf{s}'_{2,3}$ have been brought together into two complexes, instead of directly replacing all the $\mathbf{u}'_{1,j}$ monomers, they can split a $\{\mathbf{h}_2, \mathbf{h}_2^*, \mathbf{h}_2^*\}$ complex, and the $\mathbf{h}_2$ monomer can replace the $\mathbf{u}'_{1,j}$ monomers.\\Note that in this image, the only way to propagate the signal efficiently would be to use the translator gadget; not using it will be one unit of entropy less efficient, requiring 3 splits and 4 merges, showing that this TBN no longer has an entropy gap. If instead we hadn't used the translator gadgets in the previous layer, then all six $\mathbf{s}'_{2,j}$ monomers in this image would be together in a single complex rather than on two separate complexes, in which case it would be equally efficient to either use this translator gadget or directly displace the $\mathbf{u}'_{1,j}$. }
    \label{fig:translator_converging}
\end{figure}

The TBN $T_{n,k}^\mathrm{a}$ will, in the process of amplifying the signal of the analyte, form a single polymer of exponential size. This isn't an issue in the theoretical TBN model, but it is a practical issue because there is no way to design these monomers so that this large polymer would form.\footnote{The binding graph of the monomers within this giant polymer contains many complete $k$-ary trees of depth $n$ as subgraphs. If each of the nodes of this graph is a real molecule that takes up some volume, it will be impossible to embed the whole graph within 3-dimensional space as $n$ grows.}

This can be solved by adding ``translator gadgets''. These gadgets' job is to mediate between consecutive layers. Instead of monomers from one layer directly binding to monomers from the next layer, they can split apart these translator gadgets with half of the gadget going to each layer. In exchange, the TBN will no longer have exactly one stable configuration when the analyte is present, as in the TBN model, the use of these translator gates will be purely ``optional''.

We define a new TBN $\widetilde{T}_{n,k}$ (as well as $\widetilde{T}^\mathrm{a}_{n,k}$, which is obtained by adding the analyte $\mathbf{a}$). We start with the TBN ${T}_{n,k}$. To assist with the amplification step, we add monomer types $\mathbf{g}_i$ and $\mathbf{g}_i^*$ for each $2 \leq i \leq n$. Each $\mathbf{g}_i$ consists of one copy of each unstarred domain $(i,j,\ell)$ for each $1 \leq j, \ell \leq k$. Each $\mathbf{g}_i^*$ consists of the same domains but all starred. Each of these monomers has $2^{i-1}$ copies. The use of these gadgets can be seen in \cref{fig:translator_amplifying}.

To assist with the convergence step, we add monomer types $\mathbf{h}_i$ and $\mathbf{h}_i^*$ for each $2 \leq i \leq n+1$. Each $\mathbf{h}_i$ has \emph{two} copies of each unstarred domain $(i,j,\ell)'$ for each $1 \leq j, \ell \leq k$. Each $\mathbf{h}_i^*$ has only one copy of each of the corresponding starred domains. There are $2^{i-1}$ copies of each $\mathbf{h}_i$ and $2^i$ copies of each $\mathbf{h}_i^*$. The use of these gadgets can be seen in \cref{fig:translator_converging}.

\begin{thmrep}
    \label{thm:translatorgadgets}
    Let $\widetilde{T} = \widetilde{T}_{n,k}$ and $\widetilde{T}^\mathrm{a} = \widetilde{T}^\mathrm{a}_{n,k}$ be as described. Then:
    \begin{enumerate}
        \item $\widetilde{T}$ has exactly one stable configuration $\widetilde{\sigma}_{n,k}$, and $d(\widetilde{T}, \widetilde{T}^\mathrm{a}) > 2^n$. 
        \item $\widetilde{T}$ has an entropy gap of $\frac k 2$, and $\widetilde{T}^\mathrm{a}$ has the property that all of its configurations $\alpha$ that are within distance to stability $\frac k 2$ satisfy $d(\widetilde{\sigma}_{n,k}, \alpha) > 2^n$. 
        \item $\widetilde{T} = \widetilde{T}_{n,k}$ uses $\mathcal O(nk)$ total monomer types, $\mathcal O(nk^2)$ domain types, and $\mathcal O(k^2)$ domains per monomer.
        \item The unique stable configuration of $\widetilde{T}$ has $\mathcal O(k)$ monomers in its largest polymer. There is a stable configuration of $\widetilde{T}^\mathrm{a}$ sharing this property. 
    \end{enumerate}
\end{thmrep}

Compared to \cref{thm:mainthm}, this theorem trades away the condition that both TBNs have only a single stable configuration in exchange for the post-analyte TBN having a configuration with $\mathcal O(k)$ monomers per polymer, whereas the previous construction has roughly $k \cdot 2^n$ monomers in a single polymer.

The second condition is somewhat complex. This complexity's necessity is explained by \cref{fig:translator_converging}. In that figure, if we propagate signal without using the translator gadget, we arrive at a configuration that is saturated but has only one fewer complex than a stable configuration. However, such near-stable configurations are still very different from the stable configuration of $\widetilde{T}_{n,k}$, so it is still possible to distinguish the two TBNs with an amplification factor proportional to $2^n$ and a resilience to false positives and negatives proportional to $k$.

The proof of this theorem is very similar to that of \cref{thm:mainthm}, and is also left to the appendix.

\begin{proof}
    Recall the constructions of $\widetilde{T}_{n,k}$ and $\widetilde{T}^{\mathrm{a}}_{n,k}$ from \cref{subsec:translators}. Our argument will be very similar to that of \cref{thm:mainthm} (i.e., the above lemmas), except we need to account for the extra monomer types. 
    
    First, consider $\widetilde{T}_{n,k}$, where $\mathbf{a}$ is absent. We wish to show that its stable configuration looks like that of $T_{n,k}$, with the added $\mathbf{g}$ and $\mathbf{h}$ monomers only binding to added $\mathbf{g}^*$ and $\mathbf{h}^*$ monomers respectively. We order the merges to get to a saturated configuration in essentially the same order as we did in analyzing $T_{n,k}$: first we will make all merges necessary to cover all $(1, j, \ell)^*$ sites, then $(2, j, \ell)^*$, and so on up to $(n+1, j, \ell)^*$, then $(n, j, \ell)'^*$, and so on. As before, at each step, we will see that we cannot make merges in any way other than those in the desired stable configuration without needing $k - 1$ extra merges for that step. 

    For $(i, j, \ell)^*$ sites, at each step, there is exactly one way to cover all starred sites by making one merge per monomer with these starred sites: we cover each $\mathbf{s}_{i,j}$ with a $\mathbf{u}_{i,j}$ and each $\mathbf{g}_i^*$ with a $\mathbf{g}_i$. In particular, we already know from the proof of \cref{lem:pre-analyte-one-stable-config} that this is true for the $\mathbf{s}_{i,j}$ if we only use $\mathbf{s}_{i-1,j}$ and $\mathbf{u}_{i,j}$ to cover it, and that we will otherwise need to make $k - 1$ extra merges. Clearly we also cannot cover $\mathbf{g}_i^*$ with anything other than $\mathbf{g}_i$ without making $k$ merges to cover it (and thus $k - 1$ extra merges), so we cannot use $\mathbf{g}_i$ to cover $\mathbf{s}_{i,j}$.

    Likewise, for $(i, j, \ell)'^*$ sites, we only need to observe that each $\mathbf{h}_i^*$ monomer can only be covered in a single merge by $\mathbf{h}_i$, so any other way of making merges necessarily involves $k - 1$ extra merges. So by the same argument as in \cref{lem:pre-analyte-one-stable-config} and \cref{cor:pre_analyte_entropy_gap}, $\widetilde{T}_{n,k}$ has exactly one stable configuration with an entropy gap of $\frac k 2$. This configuration has $1 + \frac k 2$ monomers in the polymer containing $\mathbf{p}^*$ and all the $\mathbf{p}_j$, 3 monomers in each $\{\mathbf{h}_i, \mathbf{h}_i^*, \mathbf{h}_i^*\}$ polymer, and 2 monomers in each other polymer.

    Now, consider $\widetilde{T}^\mathrm{a}_{n,k}$, where $\mathbf{a}$ is present. If we take the stable configuration of $T_{n,k}^\mathrm{a}$ and simply put all $\mathbf{g}$ monomers into $\{\mathbf{g}_i, \mathbf{g}_i^*\}$ polymers, and all the $\mathbf{h}$ monomers into $\{\mathbf{h}_i, \mathbf{h}_i^*, \mathbf{h}_i^*\}$ polymers, we have still made exactly one merge per monomer with any starred sites, so by \cref{cor:stable-if-requires-exactly-s-merges} it is stable. If we then carry out the shifts described in \cref{fig:translator_amplifying} and \cref{fig:translator_converging}, an equal number of merges and splits are made at each step, so the resulting saturated configuration is still stable. Additionally, in this configuration, the largest polymers have $k + 3$ monomers (specifically, those containing a set of $\mathbf{s}_{i,j}$ along with one copy of $\mathbf{g}_i$ and two copies of $\mathbf{g}_{i+1}^*$). 

    All that remains to show is that all configurations of $\widetilde{T}^\mathrm{a}_{n,k}$ that are within $\frac k 2$ distance to stability have exponentially many different polymers from the stable configuration of $\widetilde{T}_{n,k}$. We will do this by showing that all $\mathbf{u}$ and $\mathbf{u}'$ monomers are free in all such configurations. 
    
    Again, this argument is very similar to the argument without the translator gadgets in \cref{lem:post_analyte_entropy_gap}. We consider merges to cover starred sites in the opposite order of the above argument for $\widetilde{T}_{n,k}$. First, consider the merges necessary to cover all the $(1,j,\ell)'$ starred sites (on $\mathbf{p}^*$). Like before, they must be covered by either all the $\mathbf{u}'_{1,j}$ monomers or all the $\mathbf{p}_j$ monomers, but using the latter gives a feed-forward configuration in which $\frac k 2 - 1$ extra merges have already been made. Thus, to be within $\frac k 2$ distance to stability, we must use the $\mathbf{s}'_{1,j}$. Next, for the $(2, j, \ell)'$ starred sites, with the merges already made, there are two copies of each of these sites all together on the polymer containing all the $\mathbf{s}'_{1,j}$, and one copy of each site on each of the two $\mathbf{h}_2^*$ monomers. If we are to merge any $\mathbf{u}'_{1,j}$ monomers to any of these in such a way that they cannot be split off without the result still being saturated, then we must merge all of the $\mathbf{u}'_{1,j}$ into one polymer. Like with the argument for $\widetilde{T}^\mathrm{a}_{n,k}$, we see that this results in $k - 1$ extra merges compared to a stable configuration. Thus, we cannot use any $\mathbf{u}'_{1,j}$, and these sites must be covered by the $\mathbf{s}'_{2,j}$ and $\mathbf{h}_2$ monomers. 
    
    We may do this either by using the $\mathbf{h}_2$ to cover both $\mathbf{h}_2^*$ (in effect, not using the translator gadget) or by using $\mathbf{h}_2$ to cover all the $\mathbf{s}'_{1,j}$. The only difference in terms of the argument is that in the former case \emph{all} of the $\mathbf{s}_{2,j}$ will be brought together in a single polymer, and in the latter case they will be split between two polymers. In the former case, it may require one extra merge to use translator gates in the next layer; however, either way, the same argument on each other layer in sequence shows that we cannot use any $\mathbf{u}'_{i,j}$ monomers without suffering $\frac k 2 - 1$ extra merges. Likewise, the exact same argument shows that the same thing is true of $\mathbf{u}_{i,j}$ monomers, necessitating that in any configuration that makes fewer than $\frac k 2 - 1$ extraneous merges, all exponentially many $\mathbf{u}$ and $\mathbf{u}'$ monomers must be free, as desired.
\end{proof}

\section{Upper Limit on TBN Signal Amplification}
\label{sec:upperlimit}
In this section, we show the following theorem providing an upper bound on the distance between a TBN before and after adding a single copy of a monomer, showing that the distance is at most double-exponential in the ``size'' of the system: 

\newcommand{\thmUpperlimit}{
    \label{thm:upperlimit}
    Let $T$ be a TBN with $d$ domain types, $m$ monomer types, and at most $a$ domains on each monomer. Let $n = \max\{d, m, a\}$. Let $T'$ be $T$ with one extra copy of some monomer. Then $d(T, T') \leq n^{8n^{7n^2}}$.
}

% This proof should never be inline. It's designed to be read after this section

\opt{append,inline}{
    \begin{thmrep}
        \thmUpperlimit
    \end{thmrep}
}

\opt{dnafinal,strip}{
    \begin{thm}
        \thmUpperlimit
    \end{thm}
}

Recall \cref{def:distance} for the distance between TBNs. Essentially, this theorem is saying that adding a single copy of some monomer can only impact doubly exponentially many total polymers, no matter how many total copies of each monomer are in the TBN. 

Our strategy for proving this theorem is to fix some ordering on polymer types, and bound the distance between the lexicographically earliest stable configuration of an arbitrary TBN under that ordering before and after adding a single copy of some monomer. To bound this distance, we cast the problem of finding stable configurations of a TBN as an integer program (IP), and use methods from the theory of integer programming value functions to give a bound on how much the solution to this IP can change given a small change in the underlying TBN.

\begin{proof}
\opt{append,inline}{
\label{proof:upperlimit}
Recall we are trying to sequentially bound the difference between the amounts $x_\mathbf{P}$ of polymers in the lexicographically earliest configurations of $T$ and $T'$, an arbitrary TBN before and after adding a copy of a monomer $\mathbf{a}$. We defined a sequence of IPs in \cref{eqn:IP_variable_comparison} whose optimal values give the amounts of polymers in these configurations. Throughout we will be loose with coefficients as our main concern is showing that the bound is doubly exponential.

For those $x_\mathbf{P}$ representing polymers in only one polymer basis, we already know that their difference is bounded by $|\mathbf{a}|$. To be conservative, we will both assume that we must account for this difference for all $P$ variables, and then carry out the rest of the analysis as though we must account for all $P$ variables being in $\pbasis(T) \cap \pbasis(T')$ (as these latter variables will actually contribute more to the analysis).

Under these assumptions, we can bound the difference between $P_{tot}$ (the total number of polymers in a stable configuration) for the two TBNs: $P_{tot}$ is the optimal value of the IP in \cref{eqn:IP} when adding in constraints that all these variables that are only in one of the two polymer bases have specific values (which are either 0, or bounded by $|\mathbf{a}|$). \cref{eqn:IP} differs between the two cases by these $P$ constraints differing by up to $|a|$ and one additional constraint $T(\mathbf{a})$ differing by 1 because of the one extra copy of $\mathbf{a}$. Therefore, by \cref{eqn:upper_limit_to_ip_change}, $P_{tot}$ differs between the two by at most $M_1(P|\mathbf{a}| + 1) + K_3$.

Now, we account for all variables that are in $\pbasis(T) \cap \pbasis(T')$. For the first such $x_\mathbf{P}$, we see that $P_{tot}$ may have changed by $M_1(P|\mathbf{a}| + 1) + K_3$, $T(\mathbf{a})$ has changed by a fixed 1, and up to $P$ of the variables representing polymers in only one polymer basis may have been fixed in value. Thus, we can bound the difference between the norm of the right-hand side constraint vectors in the two versions of \cref{eqn:IP_variable_comparison} for this first such variable by:
\begin{equation}
    P|\mathbf{a}| + M_1(P|\mathbf{a}| + 1) + K_3 + 1 \leq 2M_1P|\mathbf{a}| + K_3.
\end{equation}
It follows by \cref{eqn:upper_limit_to_ip_change} that the difference between the value of this first $x_\mathbf{P}$ before and after adding one copy of a monomer is bounded by $M_1(2M_1P|\mathbf{a}| + K_3) + K_3$. Then for the next variable in order, the value of $x_\mathbf{P}$ is baked in as the right-hand side of a constraint, meaning that this difference now contributes to the value $\left|\left|v - w\right|\right|$ in \cref{eqn:upper_limit_to_ip_change}. Thus, if $\mathbf{P}_i$ denotes the $i$th polymer in our ordering, we obtain a recurrence relation yielding a bound $B_i$ on the difference between $x_{\mathbf{P}_i}$ before and after adding $\mathbf{a}$, for $1 \leq i \leq P$ (where $B_0$ is defined as a base case):
\begin{equation}
\begin{split}
    &B_0 = 2M_1P|\mathbf{a}| + K_3\\
    &B_i = K_3 + M_1\sum_{j=0}^{i-1}B_{j}
\end{split}
\end{equation}

We can bound $\sum_{j=0}^{i-1}B_{j}$ by $2B_{i-1}$ as this sequence clearly grows faster than $2^i$, which allows us to solve the recurrence to see that all terms of $B_i$ are subleading to $K_3 \cdot (2M_1)^i$ and each term can be safely bounded by $2K_3 \cdot (2M_1)^i$. 

Thus, the total distance between these stable configurations is bounded by the sum of $P$ terms that each have this bound when we replace $i$ with $P$, giving a bound of $2PK_3 \cdot (2M_1)^P$.

We now bound these individual values. Let $S$ be the maximum number of monomers in any polymer. The value $K_3$ in \cite{blair1982ip} is constructed from three other values as the expression $M_1M_3 + M_2$. These values can each be bounded, some based on another constant $K_2$ which will be bounded next:
\begin{enumerate}
    \item $M_1$ is essentially a bound on how quickly the objective value of the corresponding real-valued linear program can change as the right-hand side changes, as described earlier. The optimal value of a linear program is always one of some set of $\lambda_i \cdot v$ where the $\lambda_i$ vectors come from the extreme points of a polyhedron based on the IP; \cite{blair1982ip} bounds $M_1$ by the maximum norm of these $\lambda_i$. A bound in our case is $P$: if there was a $\lambda_i$ whose elements summed to more than $P$, then this would imply we could get a configuration with more total polymers (as the sum of polymer counts is our objective function) than monomers, which is impossible. 
    \item $M_2$ is the maximum value of the objective function when all variables are at most $K_2$, which in this case is $PK_2$. 
    \item $M_3$ is the maximum norm of the constraint matrix times the vector of variables when all variables are at most $K_2$. In TBN language, this gives the number of monomers present in a configuration with $K_2$ copies of every polymer in the polymer basis. The elements of the constraint matrix are bounded by $S$, so $M_3$ is bounded by $SPK_2$.
\end{enumerate}

Thus, $K_3 = M_1M_3 + M_2 \leq 2SP^2K_2$. The value of $K_2$ is constructed by taking subsets $B$ of variables with the following property: the space of $a_j$ such that $\sum a_jx_j = 0$ is one-dimensional. In TBN language, an element of this subspace corresponds to a pair of configurations $\alpha$ and $\beta$ of some TBN such that the polymer types present in $\alpha$ and the polymer types present in $\beta$ are disjoint subsets of $B$. In other words, it represents a way to take a configuration using polymer types in $B$ and reconfigure it to use entirely different polymer types in $B$. This can be seen by letting positive $a_j$ give counts of polymers in $\alpha$ and negative $a_j$ give counts of polymers in $\beta$. The value $K_2$ is as large as the greatest number of a single polymer type that may be necessary for such a reconfiguration.

Then to bound this, we observe that we can find the solutions to this homogeneous system of equations by doing Gaussian elimination on an at most $n \times P$ (as there are at most $n$ monomer types) matrix whose elements are bounded by $S$. A simple bound for the largest element that can occur in this Gaussian elimination is $S^{n+1}$, so this is also a bound on our constant $K_2$. 

Thus, $K_3 \leq 2S^{n+2}P^2$. 

Finally, we bound $S$ and $P$. \cite{doty2017tbn} shows that for a TBN with $d$ domain types, $m$ monomer types and $a$ domains per monomer, the largest polymer in any stable configuration has size at most $2(m+d)(ad)^{2d+3}$. Since $n$ is a uniform bound on all these values, the maximum size of any polymer is $4n \cdot (n^{4n+6}) \leq n^{5n}$. That is to say, $S \leq n^{5n}$. Since there are at most $n$ different monomer types, the number of polymers of a given size $i$ that can be formed out of them is at most $i$ multichoose $n$ = $\binom{i + n - 1}{n} \leq (i+n-1)^n$. Therefore, the total number of possible polymer types in the polymer basis is bounded by:

\begin{equation}
    \begin{split}
        P &\leq \sum_{i = 1}^{S} (i + n - 1)^n\\
        &\leq S \cdot (S + n - 1)^n\\
        &\leq n^{5n}(n^{5n} + n - 1)^n\\
        &\leq n^{6n^2}.
    \end{split}
\end{equation}

Thus, we finally obtain the following bound on the distance between these configurations $\sigma$ and $\sigma'$:

\begin{equation}
    \begin{split}
        d(\sigma, \sigma') &\leq 2PK_3 \cdot (2M_1)^P\\
        &\leq 4S^{n+2}P^3 \cdot (2P)^P\\
        &\leq n^{6n^2}P^3 \cdot (2P)^P\\
        &\leq n^{24n^2} \cdot (2n^{6n^2})^{n^{6n^2}}\\
        &\leq n^{24n^2} \cdot (n^{7n^{7n^2}})\\
        &\leq n^{8n^{7n^2}}.
    \end{split}
\end{equation}
}
\opt{dnafinal,strip}{
A complete proof including all technical details can found in the full version of this paper on arxiv. Here we present only the main ideas of the proof.
}
\end{proof}

We first introduce a definition from \cite{haley2020tbn} and some notation that was unnecessary in previous sections. 

\begin{defn}
    \label{def:polymerbasis}
    Given a (star-limiting) TBN $T$, the \emph{polymer basis} of $T$, denoted $\pbasis(T)$, is the set of polymers $\mathbf{P}$ such that both of the following hold:
    \begin{itemize}
        \item $\mathbf{P}$ appears in some saturated configuration of a star-limiting TBN using the same monomer types as $T$.
        \item $\mathbf{P}$ cannot be split into two or more self-saturated polymers.
    \end{itemize}
\end{defn}

The polymer basis is a useful construction because it is known to describe exactly those polymer types that may appear in stable configurations of $T$. It is always finite, and we will bound its size later.

Given a TBN $T$, let $\mtypes(T)$ denote its monomer types, and let $T(\mathbf{m})$ denote the count of monomer $\mathbf{m}$ in $T$. Given a polymer $\mathbf{P}$ and a monomer type $\mathbf{m} \in \mtypes(T)$, let $\mathbf{P}(\mathbf{\mathbf{m}})$ represent the count of monomer $\mathbf{m}$ in polymer $\mathbf{P}$. 

Suppose for the rest of this section that we have a TBN $T$ to which we wish to add a single copy of some monomer $\mathbf{a}$ (which may or may not exist in $T$). Let $T'$ be $T$ with $\mathbf{a}$ added. 
% Our goal is to construct some stable configurations $\sigma$ of $T$ and $\sigma'$ of $T'$ such that we can upper-bound $d(\sigma, \sigma')$. Note that we have to be somewhat careful in choosing $\sigma$ and $\sigma'$ to give us a bound that is independent of the monomer counts. Indeed, even a single TBN can easily have multiple stable configurations that are arbitrarily far from each other as the number of copies of each monomer increases. For example, a TBN with arbitrarily many copies each of $\{a\}, \{a, b\}$, and $\{a^*\}$ could have configurations that cover the $\{a^*\}$ monomers with any combination of $\{a\}$ and $\{a, b\}$ monomers. 

\subsection{Finding Stable Configurations via Integer Programming}

Prior work~\cite{haley2020tbn} has shown that the problem of finding the stable configurations of a TBN can be cast as an IP. There are multiple different formulations; we will use a formulation that is better for the purpose of reasoning theoretically about TBN behavior. 

Let $\{x_\mathbf{P} : \mathbf{P} \in \pbasis(T)\}$ be variables each representing the count of polymer $\mathbf{P}$ in a configuration of $T$. Then consider the following integer programming problem:

\begin{equation}
\label{eqn:IP}
\begin{split}
    \text{max}& \quad &&\sum_{\mathbf{P} \in \pbasis(T)} x_\mathbf{P} \\
    \text{s.t.}& \quad &&\sum_{\mathbf{P} \in \pbasis(T)} \mathbf{P}(\mathbf{m}) x_\mathbf{P} = T(\mathbf{m}) \quad &&&\forall \mathbf{m} \in \mtypes(T) \\
    &&&x_\mathbf{P} \in \mathbb{N} \quad &&&\forall \mathbf{P} \in \pbasis(T)
\end{split}
\end{equation}

Intuitively, the linear equality constraints above express ``monomer conservation'':
the total count of each monomer in $T$ should equal the total number of times it appears among all polymers. 

The following was shown in \cite{haley2020tbn}; for the sake of self-containment, we show it here as well

\begin{prop}
\label{prop:IP-correspondence}
The optimal solutions to the IP~\eqref{eqn:IP} correspond exactly to stable configurations of $T$.
\end{prop}

\begin{proof}
If the variables $x_\mathbf{P}$ form a feasible solution, then those counts of polymers are a valid configuration because they exactly use up all monomers. If the solution is optimal, then there is no saturated configuration with more polymers (as only polymers from $\pbasis(T)$ can show up in stable configurations), so the configuration is stable. Conversely, if a configuration $\sigma$ is stable then it can be translated into a feasible solution to the IP because it only uses polymers from $\pbasis(T)$ and obeys monomer conservation. If there were a solution with a greater objective function, then this would translate to a configuration with more complexes that is still saturated (because all polymers in the polymer basis are self-saturated), contradicting the assumption of $\sigma$'s stability.
\end{proof}

We observe that adding an extra copy of some monomer to a TBN corresponds to changing the right-hand side of one of the constraints of this IP by one. Note that this is true even if we add a copy of some monomer for which there were 0 copies, as we may still include variables for polymers that contain that monomer in the former IP and simply consider there to be 0 copies of the monomer. Therefore, we are interested in sensitivity analysis of how quickly a solution to an IP can change as the right-hand sides of constraints change. 

However, there is one edge case we must account for first. It is possible that $T$ and $T'$ have different polymer bases. This is because of the first requirement in \cref{def:polymerbasis} requiring that the polymer basis respects that starred sites are limiting. If we add a single copy of a monomer, this may change which sites are limiting, if $\mathbf{a}$ has more copies of a starred site than $T$ had excess copies of the unstarred site. We cannot include variables for such polymers in the IP formulation without taking extra precautions, as if we do there may be optimal solutions that don't correspond to saturated configurations. Therefore, we will first account for how many copies of such a polymer $T$ and $T'$ may differ by:

\begin{lem}
    \label{lem:what_if_starred_sites_flip}
    Suppose that some polymer $\mathbf{P}$ is exactly one of $\pbasis(T')$ and $\pbasis(T)$. Then any saturated configuration of $T'$ contains at most $|\mathbf{a}|$ copies of $\mathbf{P}$, where $|\mathbf{a}|$ denotes the number of sites on $\mathbf{a}$. 
\end{lem}

Note that this result is slightly surprising---one natural way that one might try to design a TBN that amplifies signal is by designing the analyte so that it intentionally flips which sites are limiting. This result shows that this is an ineffective strategy: going from 5 excess copies of some site $a$ to 5 excess copies of $a^*$ is seemingly no more helpful in instigating a large change than going from 60 excess copies of $a$ to 50.

\begin{proof}
    If $\mathbf{P}$ is in $\pbasis(T')$ but not $\pbasis(T)$, it must contain an excess of a starred site that was limiting in $T$, but is no longer limiting in $T'$. We see this because $\mathbf{P}$ necessarily occurs in a saturated configuration of the TBN containing precisely the monomers that it is composed of; therefore, in order to not be in $\pbasis(T)$, by definition of the polymer basis, it must be the case that this TBN has different limiting sites than $T'$. 

    Let $a$ denote some such site type, so that $a^*$ is limiting in $T$ and $a$ is limiting in $T'$, and $\mathbf{P}$ contains an excess of $a^*$. Then $\mathbf{a}$ must contain an excess of $a^*$, but it cannot contain more than $|\mathbf{a}|$ excess copies. Therefore, there are at most this many \emph{total} excess copies of $a^*$ in $T'$. It follows that if there are more than $|\mathbf{a}|$ copies of $\mathbf{P}$ in a configuration of $T'$, then those copies of $\mathbf{P}$ collectively have more excess copies of $a^*$ than $T'$ does, so some other polymer in that configuration would have to have an excess of $a$. This implies that such a configuration is not saturated (and therefore also cannot be stable). An identical argument shows that the same is true for polymers in $\pbasis(T)$ but not $\pbasis(T')$.  
\end{proof}

In order to analyze and compare the two IP instances, we need them to have the same variable set. Therefore, we will include variables for all polymers from both polymer bases in both IP formulations. Let $\pbasis(T, T') = \pbasis(T) \cup \pbasis(T')$ denote this merged polymer basis, and let $P = |\pbasis(T, T')|$ denote the total number of possible polymers we must consider, or equivalently the number of variables we will have in these IPs. In each IP, we will have a constraint on each variable representing a polymer not in the relevant polymer basis, that says that that variable must equal zero.  

\subsection{Sensitivity Analysis}

This sensitivity analysis problem of how IPs change as the right-hand sides of constraints change was studied by Blair and Jeroslow in \cite{blair1982ip}. We will not need their full theory, but we will use some of their results and methods.

In Corollary 4.7 of \cite{blair1982ip}, they show that there is a constant $K_3$, independent of the right-hand sides of constraints (in our case, independent of how many copies of each monomer exist) such that:

\begin{equation}
    R_c(v) \leq G_c(v) \leq R_c(v) + K_3,
\end{equation}
where $G_c(v)$ gives the optimal value of the objective function $c$ of a minimization IP as a function of the vector $v$ of right-hand sides of constraints, and $R_c(v)$ gives the optimal value of the same problem when relaxing the constraint that variables must have integer values. The objective function we've shown so far is to maximize the sum of polymer counts rather than minimize, but the same statement applies that the integer and real-valued optimal solutions differ by at most $K_3$. In defining $K_3$, they also show the existence of a constant $M_1$ such that 
\begin{equation}
    |R_c(v) - R_c(w)| \leq M_1\left|\left|v - w\right|\right|,
\end{equation}
where $v$ and $w$ are different vectors for the right-hand sides of constraints. Note that we take all norms as 1-norms. Combining these inequalities, we see that
\begin{equation}
    \label{eqn:upper_limit_to_ip_change}
    |G_c(v) - G_c(w)| \leq M_1\left|\left|v - w\right|\right| + K_3.
\end{equation}
For example, if we want to know the difference between the total number of polymers in a stable configuration before and after adding one copy of a monomer (and if the polymer bases of $T$ and $T'$ are identical), then we care about increasing one element of $v$ by 1, so our bound on this difference is $M_1 + K_3$. This statement applies to maximization and minimization problems.

\subsection{From Optimal Values to Polymer Counts}

For ease of analysis, we order the polymers in $\pbasis(T, T')$ as follows: first we list all the polymers that are not in $\pbasis(T)$, then all the polymers that are not in $\pbasis(T')$, then all the polymers in $\pbasis(T) \cap \pbasis(T')$. We need to show that the number of copies of each individual polymer does not change too much. We do this using a technique similar to Corollary 5.10 in \cite{blair1982ip}.

Let $P_{tot}$ be the total number of polymers in a stable configuration (either before or after adding $\mathbf{a}$, depending on which case we are examining). 

We now define a new sequence of integer programs whose optimal values give polymer counts in the lexicographically earliest stable configuration under this ordering. We do this by finding the value of each variable $x_\mathbf{P}$ in order. This sequence of IP problems is defined separately for both TBNs, before and after adding $\mathbf{a}$. 

For those variables representing a polymer that is in one basis but not the other, we do not need to analyze this IP, so we simply fix such a variable's value to whatever its value is in this lexicographically earliest stable configuration, which will be 0 in one TBN and bounded by $|\mathbf{a}|$ by \cref{lem:what_if_starred_sites_flip} in the other.

Now, to find the value of some particular variable $x_\mathbf{Q}$ in either of the two TBNs where $\mathbf{Q} \in \pbasis(T) \cap \pbasis(T')$, suppose we have already found the value $y_\mathbf{P}$ we wish to fix $x_\mathbf{P}$ to for each $\mathbf{P} < \mathbf{Q}$ under our ordering. Then we define a new IP on all the same variables as follows:

\begin{equation}
\label{eqn:IP_variable_comparison}
\begin{split}
    \text{min}& \quad &&x_\mathbf{Q} \\
    \text{s.t.}& \quad &&\sum_{\mathbf{P} \in \pbasis(T, T')} x_\mathbf{P} = P_{tot} \\
    & \quad &&\sum_{\mathbf{P} \in \pbasis(T, T')} \mathbf{P}(\mathbf{m}) x_\mathbf{P} = T(\mathbf{m}) \quad &&&\forall \mathbf{m} \in m(T) \\
    & \quad &&x_\mathbf{P} = y_\mathbf{P} \quad &&&\forall \mathbf{P} < \mathbf{Q} \\
    &&&x_\mathbf{P} \in \mathbb N \quad &&&\forall \mathbf{P} \in \pbasis(T, T')
\end{split}
\end{equation} 

By construction, this IP gives us the smallest possible value that $x_\mathbf{Q}$ can take on in a stable configuration (as all variables must sum to $P_{tot}$) in which all previous $x_\mathbf{P}$ have fixed values. Then this process gives us a sequence of $P$ (minus however many polymers were only in one polymer basis) different pairs of IP problems that we can sequentially compare to bound the differences between the values of the individual polymer counts in these lexicographically earliest configurations. We can repeatedly apply \cref{eqn:upper_limit_to_ip_change} to each $x_\mathbf{P}$ in turn, as each variable's value before and after adding $\mathbf{a}$ will be given by the optimal value of (\ref{eqn:IP_variable_comparison}) where the only differences are in the right-hand sides of constraints. 
\opt{append,inline}{The remaining proof of \cref{thm:upperlimit} consists mostly of making these bounds concrete, and is left to the appendix.}

\section{Conclusion}

In this paper we have defined the signal amplification problem for Thermodynamic Binding Networks, and
we have demonstrated a TBN that achieves exponential signal amplification.
We also showed a doubly-exponential upper bound for the problem.
As TBNs model mixtures of DNA, a TBN that amplifies signal can potentially be implemented as a real system. An upper bound has implications for how effective a system designed in this way can potentially be, and shows that there are some limitations for a purely thermodynamic approach to signal detection and amplification. 

One clear direction for future work is to implement such a system. This would involve creating a design that accounts for the simplifications of the TBN model.
In particular, 
enthalpy and entropy need to be strong enough with enthalpy sufficiently stronger than entropy.
Further, the polymers formed need to be geometrically feasible. 
We have done some work to make this problem geometrically realizable with the inclusion of translator gadgets in \cref{subsec:translators}. In principle, the polymers that are formed in this version of the system are simple enough that they should form if the DNA strands implementing them are well-designed. 

Another goal would be to bridge the gap between our singly exponential amplifier and doubly exponential upper bound by either describing a TBN that can amplify signal more than exponentially, or deriving a more precise upper bound. If one wished to construct a TBN with doubly exponential amplification, an examination of our upper bound proof will show that such a TBN must have an exponentially sized polymer basis, and most likely would need to actually use an exponential amount of different polymer types in its stable configurations either with or without the analyte. Such a design seems relatively unlikely to come to fruition, and it seems more likely that our proof technique or similar techniques can be tightened in order to show a stricter upper bound. Thus, we conjecture that the true upper bound is (singly) exponential.

There are also other types of robustness that we have not discussed in this work that merit further analysis. One of these is input specificity: the question of how well the system amplifies signal if the analyte is changed slightly. Another is sensitivity to the number of copies of each component. Intuitively, our system's behavior depends on having exactly equal numbers of complementary strands within each layer; if there are too many copies of one, it may result in those excess copies spuriously propagating or blocking signal to the next layer. This issue may be intrinsic to thermodynamic signal amplifiers, or there may be some system more robust to it. Lastly, it may be experimentally useful to show that our system achieves its stable states not only in the limit of thermodynamic equilibrium, but also more practically when annealed. Some systems such as HCR are designed to reach non-equilibrium, meta-stable states when annealed. 
We conjecture that our system should reach equilibrium when annealed, because kinetic traps in the system are far away from being thermodynamically stable (large entropy gap). 
Formally studying annealing could be done by analyzing versions of the TBN model with different tradeoffs between entropy and enthalpy to model different temperatures.

%%
%% Bibliography
%%

%% Please use bibtex, 

\bibliography{main}

\appendix
\end{document}